\documentclass[twoside,11pt]{article}
\usepackage{jmlr2e}

\usepackage{algorithm}
\usepackage{algcompatible}
\usepackage{amsmath}
\algnewcommand\INPUT{\item[\textbf{Input:}]}%
\algnewcommand\OUTPUT{\item[\textbf{Output:}]}%

\usepackage{amssymb}
\usepackage{graphicx,psfrag,epsf}
\usepackage{enumerate,verbatim}
\usepackage{enumitem}
\usepackage{natbib}
\usepackage{url} 
\usepackage{bm}
\usepackage[utf8]{inputenc}
\usepackage[english]{babel}
\usepackage{longtable,booktabs}
\usepackage{mathtools}
\usepackage{color}

\def\olbx{\overline{\bx}}
\def\b0{\mathbf{0}}
\def\b1{\mathbf{1}}

\def\argmin{\operatornamewithlimits{argmin}}

\def\bbeta{\boldsymbol{\beta}}

\def\cov{\mathrm{Cov}}
\def\diag{\operatorname{diag}}

\def\bD{\mathbf{D}}

\def\bI{\mathbf{I}}
\def\rI{\mathrm{I}}

\def\bJ{\mathbf{J}}

\def\logit{\mathrm{logit}}

\def\pr{\mathrm{P}}

\def\bR{\mathbf{R}}

\def\bx{\mathbf{x}}

\def\var{\mathrm{Var\,}}
\def\bV{\mathbf{V}}

\def\bX{\mathbf{X}}
\def\bY{\mathbf{Y}}
\def\bE{\mathbf{E\,}}

\def\rT{\mathrm{T}}

\usepackage{lastpage}
\jmlrheading{22}{2021}{1-\pageref{LastPage}}{2/19; Revised
12/20}{3/21}{19-132}{Zhe Fei and Yi Li}
\ShortHeadings{Estimation and Inference for High Dimensional GLMs}{Fei and Li}

\firstpageno{1}

\begin{document}

		\title{\bf Estimation and Inference for High Dimensional Generalized Linear Models: A Splitting and Smoothing Approach}
		\author{\name Zhe Fei \email feiz@ucla.edu \\
			\addr Department of Biostatistics\\
			UCLA\\
			Los Angeles, California, 90025 \\
			\AND
			\name Yi Li \email yili@umich.edu \\
			\addr Department of Biostatistics\\
			University of Michigan\\
			Ann Arbor, Michigan, 48109}
	\editor{Boaz Nadler}	
		\maketitle

	\begin{abstract}%
        The  focus of modern biomedical studies {has gradually shifted} to {explanation and estimation of} joint effects of high dimensional predictors on  disease {risks}. Quantifying  uncertainty {in} these estimates may {provide valuable insight into} prevention strategies or treatment decisions for both patients and physicians. High dimensional inference, including confidence intervals and hypothesis testing, has sparked much interest. {While much} work has been done {in the} linear regression {setting,} there is lack of literature on inference for high dimensional generalized linear models{. We} propose a novel and computationally feasible method, which accommodates a variety of  outcome types, including normal, {binomial}, and Poisson {data}.
		We use a ``splitting and smoothing'' approach, which splits samples into two parts,  performs variable selection using one part and conducts partial {regression}  with the other part. Averaging the estimates over multiple random {splits}, we obtain the smoothed estimates, which are numerically stable. We show that the estimates are consistent,  asymptotically normal, and construct confidence intervals with proper coverage probabilities for all predictors. 
	    We  examine the finite sample performance of our method by comparing it with the existing methods and applying  it to analyze a lung cancer cohort study.
		
	\end{abstract}
	
	\begin{keywords} Confidence intervals, dimension reduction, high dimensional inference for GLMs, sparsity, sure screening
	\end{keywords}

	\section{Introduction}
	\label{s:1}
	
	In {the big data era}, high dimensional regression has been widely used 
	to address questions arising from many scientific fields, ranging from genomics to sociology  \citep{hastie2009elements,fan2010selective}. 
	For example, modern biomedical research has gradually shifted to understanding  joint effects of  high dimensional predictors on  disease {outcomes} ({e.g. molecular biomarkers on the} onset of lung cancer) \citep[among others]{vaske2010inference,chen2014semi}. 
	{A motivating} clinical study is the Boston Lung Cancer Survivor Cohort (BLCSC){, one of the largest} comprehensive lung cancer survivor {cohorts}, which investigates the molecular mechanisms underlying lung cancer \citep{Christiani2017blcs}. 
		Using a target gene approach \citep{moon2003current,garrigos2018clinical,ho2019machine},  we {analyzed} a subset of  $708$ lung cancer patients and $751$ controls,   with 6,800 single nucleotide polymorphisms (SNPs) from $15$ cancer related genes, in addition to demographic variables such as age, gender, race, education level, and smoking status. Our objective {was} to {determine which} {covariates were predictive in} distinguishing cases from controls. {As} smoking {is known to play} a {significant} role in the development of lung cancer, we {were}  interested in estimating and testing the interaction between smoking status (never {versus} ever smoked) and {each SNP}, in addition to {the} main {effect of the SNP}.  Quantifying  uncertainty {of} the estimated effects helps inform prevention strategies or treatment decisions for patients and physicians \citep{minnier2011perturbation}.

	Considerable {progress} has been made in drawing inferences based on penalized linear models \citep{zhang2014confidence,javanmard2014confidence,buhlmann2014high,dezeure2015high}. {While techniques for variable selection and estimation  in high dimensional settings have been extended to generalized linear models (GLMs) and beyond \citep{van2008high,fan2009ultrahigh,witten2009covariance}, high dimensional inference in these settings is still at its infancy stage.} For example, \citet{buhlmann2014high}  generalized {the} de-sparsified LASSO to high dimensional GLMs, while \cite{ning2017general} proposed a de-correlated score test for penalized M-estimators. In the presence of high dimensional control variables, \cite{belloni2014inference,belloni2016post} proposed a post{-}double selection procedure for estimation and inference of a single treatment effect {and} \cite{lee2016exact} characterized the distribution of a post-LASSO-selection estimator {conditional} on the \emph{selected variables}, but only for the linear regression.  
	
	However, the performance of these methods may depend heavily on  tuning parameters, often chosen by computationally intensive cross-validation. Also, these methods may require inverting a $p \times p$ information matrix (where $p$ is the number of predictors), or equivalently, estimating a $p \times p$ precision matrix, with extensive computation and stringent technical conditions. For example, the sparse precision matrix assumption may be  violated in GLMs, resulting in biased estimates \citep{lu2020revisit}. 
	
	We propose {a new approach for drawing} inference with high dimensional GLMs. The idea is to randomly split the samples into two {sub-samples} \citep{meinshausen2009p},  use {the first sub-sample} to select a subset of important predictors and achieve  dimension reduction, and  use the remaining samples to {parallelly} fit low dimensional GLMs by appending each predictor to the selected set, one at a time, {to} obtain the estimated coefficient for each predictor, regardless of being selected or not.
	{As with other methods for} high dimensional regression \citep{zhang2014confidence,javanmard2014confidence,buhlmann2014high}, one key assumption is that the number of non-zero components of $\bbeta^*$ is small relative to the sample size, where $\bbeta^*$ are the {true values underlying} the  parameter vector, $\bbeta$, in a regression model.  The sparsity {condition} {is} reasonable in some biomedical applications. For example, in the context of cancer genomics, it is likely that a certain type of cancer is related to only a handful of oncogenes and tumor suppressor genes \citep{lee2010oncogenes,goossens2015cancer}. 
	Under this sparsity condition, we  show that  our proposed estimates are consistent and asymptotically normal. However, {these} estimates can be highly variable due to {both} the random {splitting} of  data and the variation incurred {through} selection. 
	To stabilize the  estimation and  account for the variation {induced by} variable selection, we repeat the random splitting a number of times and average  the resulting estimates to obtain the  smoothed estimates{. These smoothed estimators} are consistent and asymptotically normal, with improved efficiency.
	
	Our approach, termed Splitting and Smoothing for GLM (SSGLM), aligns with multi-sample splitting  \citep{meinshausen2009p, wang2020debiased} and bagging  \citep{buhlmann2002analyzing,friedman2007bagging,efron2014estimation}, and differs from the existing methods based on penalized regression \citep{zhang2014confidence,buhlmann2014high,ning2017general,javanmard2018debiasing}. The procedure has several novelties.
	{First, it addresses the} {high dimensional estimation problem through} the {aggregation of low dimensional estimations} {and  presents computational advantages over} {other existing methods. For example}, {de-biased} {methods} {\citep{buhlmann2014high,javanmard2018debiasing}} {require} {well-estimated high dimensional precision matrices} {for proper} {inference} {(e.g. correct coverage probabilities), which is  statistically and computationally challenging. Complicated procedures {that involve} choosing a large number of tuning parameters are needed to strike a balance between estimation accuracy and model complexity; see  \cite{buhlmann2014high} and  \cite{javanmard2014confidence}. In contrast, our algorithm is {more} straightforward as it avoids  estimating a high dimensional precision matrix by adopting a ``split and select'' strategy} {with minimal} tuning.  
	Second, we {have derived} the variance estimator using the infinitesimal jackknife method adapted to the splitting and smoothing procedure \citep{efron2014estimation}{. This} is free of parametric assumptions and	leads to confidence intervals with correct coverage probabilities. 
	{Third, we have relaxed the stringent ``selection consistency" assumption on variable selection, which is required in} \cite{fei2019drawing}. Our procedure is valid with a mild ``sure screening'' assumption for the selection method. 
	Finally, our framework facilitates  hypothesis testing and drawing {inference} on predetermined contrasts in the presence of high dimensional nuisance  parameters. 
	
	The rest of the paper is organized as follows. Section \ref{s2} describes the SSGLM {procedure} and Section \ref{s3} introduces {its} theoretical properties. Section \ref{s4} describes the inferential procedure and  Section \ref{s5}  extends it to accommodate any sub-vectors of parameters of interest. Section \ref{s6} provides simulations and comparisons with the existing methods. Section \ref{s7} reports {our} analysis of the BLCSC data. We conclude the paper with a brief discussion in Section \ref{s8}.
	
	\section{Method}\label{s2}
	
	\subsection{Notation}
	
	We assume the observed data  $(Y_i,\bx_i) = \left(Y_i, x_{i1},x_{i2},\ldots,x_{ip} \right), i=1, \ldots, n, $ are i.i.d.{~}copies of $(Y,\bx) = \left(Y, x_{1},x_{2},\ldots,x_{p} \right)$. Without loss of generality, we assume that the predictors are centered with $\bE(x_j)=0,\ j=1, \ldots,p.$
	In the matrix {form}, we denote the {$n$ samples of} observed data by $\bD^{(n)} =(\bY,\bX)$, where $\bY=(Y_1,\ldots,Y_n)^\rT$ and $\bX = (\bX_1,\ldots,\bX_p)$.
	Here, $\bX_j = (x_{1j}, \ldots, x_{nj})^\rT$ for $j=1, \ldots, p$. In addition, $\overline{\bX} = (\mathbf{1},\bX)$ includes {an} $n \times 1$ column vector of 1's.
	To accommodate non-Gaussian outcomes, we assume the outcome variable {belongs to the} linear exponential distribution family, which includes {the} normal, Bernoulli, Poisson, and negative{-binomial} distributions. That is, given $\bx$, the conditional density function for $Y$ is
	\begin{equation}\label{glm1}
	f(Y|\theta)  = \exp\left\{ {Y\theta - A(\theta)} +c(Y) \right\},
	\end{equation}
	where $A(\cdot)$ is a specified function that links the mean of $Y$ to $\bx$ through  $\theta$. We assume the second derivative of $A(\theta)$ is  continuous and positive. We consider the canonical mean parameter, $\theta = \olbx\bbeta$, where $\olbx = (1,\bx)$ and $\bbeta = (\beta_0,\beta_1,\ldots,\beta_p)^\rT$ {include} an intercept term. 
	Specifically, {denote} $ \mu= \bE(Y | \bx) =  A'(\theta) = g^{-1} \left(\olbx\bbeta\right)$,
	and $\bV(Y | \bx) = A''(\theta) =\nu(\mu)$, where $g(\cdot)$ and $\nu(\cdot)$ are the link and variance  functions, respectively.  
	
    The forms of $A(\cdot)$, $g(\cdot)$, {and} $\nu(\cdot)$ depend on the  data type of $Y$. For example, with the outcome in BLCSC being {a} binary {indicator of} lung cancer, 
    $A(\theta) = \log\left(1+e^{\theta} \right),\ g(\mu) = \logit(\mu) = \log\left( \frac{\mu}{1 - \mu}\right)$  and  $\nu(\mu)= \mu(1-\mu)$,
	corresponding to the well known logistic regression. Based on $(\bY,\bX)$, the negative log-likelihood with model (\ref{glm1}) is
	\begin{equation*}
	\ell(\bbeta) = \ell(\bbeta;\bY,\bX) = \frac{1}{n}\sum_{i=1}^{n} \left\{  A(\theta_i) - Y_i\theta_i \right\} = \frac{1}{n}\sum_{i=1}^{n}  \left\{  A( \olbx_i\bbeta ) - Y_i( \olbx_i\bbeta ) \right\},
	\end{equation*}
	where $\theta_i = \overline{\bx_i}\bbeta${ and}  $\overline{\bx_i} = (1,x_{i1},x_{i2},\ldots,x_{ip}).$ The score and the observed information are
	\begin{gather*}
	U(\bbeta) = \frac{1}{n} {\overline{\bX}}^\rT\left\{ A'( \overline{\bX}\bbeta ) - \bY \right\} {\rm and} \,\,
	\widehat{I}(\bbeta) = \frac{1}{n}{\overline{\bX}}^\rT \bV \overline{\bX},
	\end{gather*}
	which are a $(p+1) \times 1$ vector and a $(p+1) \times (p+1)$ matrix, respectively. Here, $\bV = \mathrm{diag}\{\nu(\mu_1),\ldots,\nu(\mu_n)\}$ and
	$\mu_i = g^{-1}(\overline{\bx_i}\bbeta)$ for $i=1, \ldots, n$. When a univariate function such as $A'(\cdot)$ {is applied} to a vector, it operates component-wise and returns a vector of values. 
	
	We add an index set, $S\subset \{1,2,\ldots,p\}$, {to the subscripts of vectors and matrices} to index subvectors $\bx_{iS} = (x_{ij})_{j\in S}$ and $\olbx_{iS} = (1,\bx_{iS})$, and  submatrices $\bX_S=(\bX_j)_{j\in S}$ and $\overline{\bX}_S= (\mathbf{1},\bX_S)$. Moreover, we define $S_{+j} = \{j\}\cup S$ and $S_{-j} = S\setminus \{j\}$. As a convention,   let $S_{+0} =S_{-0} = S$, where ``0" corresponds to the intercept.
	
	We write $\bbeta_S = (\beta_0,\beta_j)_{j\in S}$, which always includes the intercept and  is of length $1+|S|$.
	The negative log-likelihood  for model  (\ref{glm1}) that regresses   $\bY$ on $\bX_S$ (termed {the} partial regression) is
	\begin{equation} \label{partialreg}
	\ell_S(\bbeta_S)= \ell(\bbeta_S;\bY,\bX_S)  
	= \frac{1}{n}\sum_{i=1}^{n}  \left\{  A( \olbx_{iS}\bbeta_S ) - Y_i \olbx_{iS}\bbeta_S  \right\}.
	\end{equation}
	Similarly, $U_S(\bbeta_S) = n^{-1} {\overline{\bX}_S}^\rT \left( A'( \overline{\bX}_S\bbeta_S ) - \bY \right)$ {and} $ \widehat{I}_S(\bbeta_S) =
	n^{-1} {\overline{\bX}_S}^\rT \bV_S{\overline{\bX}_S},$
	where $\bV_S = \mathrm{diag}\{ A''( \olbx_{1S}\bbeta_S),\ldots, A''( \olbx_{nS}\bbeta_S )\}$. 
	Let the {true values of } $\bbeta$ be $\bbeta^* = (\beta_0^*,\beta_1^*,\ldots,\beta_p^*)$. Define the expected information as
	$I^* = \bE \{\widehat{I}(\bbeta^*)\}$.
	Let $S^* = \left\{j \ne 0: \beta^*_j \ne 0 \right\}$ denote the active set, and {let} $s_0 = |S^*|${ be} the number of nonzero and non-intercept elements in $\bbeta^*$. When $ S \supseteq S^*$, define the ``observed" \emph{sub-information} by $\widehat{I}_S= \widehat{I}_S({\bbeta}_S^*)$, and the ``expected" sub-information by  $I_S = \bE\{ \widehat{I}_S\}$.  The latter is equal to the submatrix of $I^*$ with rows and columns indexed by $S$, which is denoted by $I^*_S$.
	
	
	\subsection{Proposed SSGLM estimator} 
	
	Under model (\ref{glm1}), we assume a sparsity condition that $s_0$ is small relative to the sample size and will be detailed in Section 3. 
	We randomly split the samples, $\bD^{(n)}$, into two parts, $\bD_1$ and $\bD_2$, with sample sizes $|\bD_1|=n_1$, $|\bD_2|=n_2$, {respectively, such that} $n_1 + n_2 = n$. {For example}, {we can consider} an equal splitting with $n_1 = n_2 = n/2$. We  apply a variable selection scheme, $\mathcal{S}_\lambda$, where $\lambda$ denotes the tuning parameters{, to $\bD_2$} to select a subset of important predictors $S\subset \{1, \ldots,p\}$, {with $|S| <n  $  for dimension reduction.} {Then using} $\bD_1 = (\bY^1,\bX^1)$, for each $j = 1,2,\dots, p$, we  fit a low dimensional GLM by regressing $\bY^1$ on $\bX^1_{S_{+j}}$, where  $S_{+j} = \{j\}\cup S$. Denote the maximum likelihood estimate (MLE) of each fitted model as $\widetilde{\bbeta}_{S_{+j} }$, and define $\widetilde{\beta}_j = \left(\widetilde{\bbeta}_{S_{+j} } \right)_j$, the element of $\widetilde{\bbeta}_{S_{+j} }$ corresponding to predictor $\bX_j$. We denote by $\widetilde{\beta}_0$  the estimator of the intercept from the model $\bY^1\sim \bX^1_{S}$.  {Thus, the one-time estimator based on a single data split is defined} as
	\begin{equation}\label{onetime}
	\begin{aligned}
	\widetilde{\bbeta}_{S_{+j} } = 
	\argmin_{\bbeta_{S_{+j}} } \ell_{S_{+j}} (\bbeta_{S_{+j}}) = \argmin_{\bbeta_{S_{+j}} } \ell (\bbeta_{S_{+j}};\bY^1,\bX^1_{S_{+j}});\\
	\widetilde{\beta}_{j} = \left(\widetilde{\bbeta}_{S_{+j} } \right)_j;\quad
	\widetilde{\bbeta} = (\widetilde{\beta}_0,\widetilde{\beta}_1,\ldots,\widetilde{\beta}_p). 
	\end{aligned}
	\end{equation}
	{In the} linear regression {setting} \citep{fei2019drawing}, $\widetilde{\beta}_j$ in (\ref{onetime}) has an explicit form,\linebreak
	$\widetilde{\beta}_j = \Big\{ ({{\bX}^1_{S_{+j}}}^\mathrm{T}{\bX}^1_{S_{+j}})^{-1}{{\bX}^1_{S_{+j}}}^\mathrm{T}{\bY}^1 \Big\}_{j}\;.$
	
	The rationale for this one-time estimator is that if the subset of important predictors, $S$, is equal to or contains the  active set, $S^*$, {then} $\widetilde{\beta}_j$ would  be a consistent estimator regardless of whether variable $j$ is selected or not \citep{fei2019drawing}. We show in Theorem \ref{thm1} that the one-time estimator is indeed  consistent and asymptotically normal {in the GLM setting.}
 
	However, {the}  estimator based on a single split is highly variable,  making it difficult to separate {true} signals from noises. {This phenomenon is analogous to using a single tree in the bagging algorithm \citep{buhlmann2002analyzing}.}
	To reduce {this} {variability}, we resort to a multi-sample splitting scheme. We randomly split the data multiple times, repeat the estimation procedure, and average the resulting estimates to obtain the smoothed coefficient estimates. Specifically, for each $b =1,2, \ldots,B$, where $B$ is large,  we randomly split the data, $\bD^{(n)}$, into $\bD_1^b$ and $\bD_2^b$, with  $|\bD_1^b|=n_1$ and $|\bD_2^b|=n_2$ such that {the } splitting proportion is ${q} = n_1/n,\;0<{q} <1$. Denote the candidate set of variables {selected} by {applying} $\mathcal{S}_\lambda$ {to} $\bD_2^b$ as $S^b$, and the estimates via (\ref{onetime}),  as $\widetilde{\bbeta}^b = (\widetilde{\beta}_0^b,\widetilde{\beta}_1^b,\ldots,\widetilde{\beta}_p^b)$.
Then the smoothed estimator, termed {the} SSGLM estimator, is defined to be
	\begin{equation} \label{smest}  
	\widehat{\bbeta} = (\widehat{\beta}_0,\widehat{\beta}_1,\ldots,\widehat{\beta}_p),\;\mathrm{where}\;
	\widehat{\beta}_j = \frac{1}{B} \sum_{b=1}^{B}  \widetilde{\beta}_j^b.
	\end{equation}
The procedure is described in Algorithm \ref{alg1}.


\begin{algorithm}
	\caption{ SSGLM Estimator\label{alg1}}
	\begin{algorithmic}[1]
		\REQUIRE A variable selection procedure denoted by $\mathcal{S}_\lambda$
		\INPUT Data $(\bY,\bX)$, a splitting proportion  ${q} \in (0,1)$, and the number of random {splits} $B$
		\OUTPUT Coefficient vector estimator $\widehat{\bbeta}$
		\FOR{$b=1,2,\ldots,B$}
		\STATE Split the samples into $\bD_1$ and $\bD_2$, with  $|\bD_1|={q} n$, $|\bD_2|=(1-{q})n$
		\STATE Apply $\mathcal{S}_\lambda$ {to} $\bD_2$ to select predictors indexed by $S\subset \{1, \ldots, p\}$
		\FOR{$j=0,1,\ldots,p$}
		 \STATE With $S_{+j} = \{j\}\cup S$, fit model (\ref{glm1}) by regressing $\bY^1$  on $\bX^1_{S_{+j}}$, where $\bD_1 = (\bY^1,\bX^1)$, and compute the MLE $\widetilde{\bbeta}_{S_{+j} }$ as in (\ref{onetime})
		 \STATE Compute $\widetilde{\beta}_{j}^b = \left(\widetilde{\bbeta}_{S^b_{+j} } \right)_j$, which is the coefficient for predictor $\bX_j$ ($\widetilde{\beta}_{0}^b$ represents the intercept)
		 \ENDFOR
		 \STATE Output  $\widetilde{\bbeta}^b= (\widetilde{\beta}_0^b,\widetilde{\beta}_1^b,\ldots,\widetilde{\beta}_p^b)$
		\ENDFOR
		\STATE Compute $\widehat{\bbeta} = (\widehat{\beta}_0,\widehat{\beta}_1,\ldots,\widehat{\beta}_p),\;\mathrm{where}\;
		\widehat{\beta}_j = \frac{1}{B} \sum_{b=1}^{B}  \widetilde{\beta}_j^b$
	\end{algorithmic}
\end{algorithm}

\begin{algorithm}
	\caption{Model-free Variance Estimator\label{alg2}}
	\begin{algorithmic}[1]
		\INPUT $n, n_1, B$, $\widetilde{\bbeta}^b, b=1,2,\ldots,B$ and $\widehat{\bbeta}$
		\OUTPUT Variance estimator $\widehat{V}^B_j$ for $\widehat{\beta}_j$, $j=0,1,\ldots,p$
		\STATE For $i=1,2,\ldots,n$ and $b=1,2,\ldots,B$, define $J_{bi} = \bI\left( (Y_i,\bx_i)\in \bD_1^b  \right) \in\{0,1\}$, and $J_{\cdot i} = \left(\sum_{b=1}^{B}J_{bi}\right)/B$		
		\FOR{$j = 0,1,\ldots,p$}
		\STATE Compute \begin{equation*} 
		\widehat{V}_j= \frac{n(n-1)}{\left(n-n_1\right)^2} \sum_{i=1}^{n}\widehat{\mathrm{cov}}^2_{ij},
		\end{equation*}
		where
		\begin{equation*}
		\widehat{\mathrm{cov}}_{ij}=\frac{1}{B} \sum_{b=1}^{B}\left( J_{bi}- J_{\cdot i}\right) \left( \widetilde{\beta}^b_j-\widehat{\beta}_j \right)   
		\end{equation*}
		\STATE Compute \begin{equation*}
		    \widehat{V}_j^B = \widehat{V}_j - \frac{n}{B^2}\frac{n_1}{n-n_1}\sum_{b=1}^{B}(\widetilde{\beta}^b_j-\widehat{\beta}_j)^2
		\end{equation*}
		\ENDFOR
		\STATE Set $\widehat{V}^B = \left(\widehat{V}^B_1,\widehat{V}^B_2,\dots, \widehat{V}^B_p \right)$
	\end{algorithmic}
\end{algorithm}
	
	\section{Theoretical Results}\label{s3}
	
We specify the following regularity conditions. 
\begin{enumerate}[label=(A\arabic*)]
	\item \label{a1} (\emph{Bounded observations}) $\|\bx\|_\infty\le C_0$ and $\bE |Y|< \infty$.  Without loss of generality, we assume $C_0 = 1$. 
	\item \label{a2} (\emph{Bounded eigenvalues and effects}) The eigenvalues of $\Sigma=\bE(\olbx^\rT \olbx)$, where $\olbx = (1,\bx)$, are bounded below and above by constants $c_{\min}, c_{\max}$, such that
	$$ 0 < c_{\min} \le \lambda_{\min}\left(\Sigma \right) < \lambda_{\max}\left(\Sigma \right) \le c_{\max} < \infty.$$ 
	In addition, there exists a constant $c_\beta>0 $ such that $|\bbeta^*|_\infty \le c_\beta$. 
	\item \label{a3} (\emph{Sparsity and sure screening property}) Recall that $S^* = \left\{j \ne 0: \beta^*_j \ne 0 \right\}$ and $s_0 = |S^*|$.
	Let $\widehat{S}_{\lambda_n}$ be the index set of predictors selected by $\mathcal{S}$ with a tuning parameter $\lambda_n$.
	Assume $\log p =o(n^{1/2})$, there exists a sequence $\{\lambda_n\}_{n\ge 1}$ and constants $0\le c_1<1/2$, $c_2, K_1,K_2>0$ such that $s_0  \le K_1n^{c_1}$,  $|\widehat{S}_{\lambda_n}|\le K_1n^{c_1}$, and 
	\begin{equation*}
	\pr\left( S^* \subseteq \widehat{S}_{\lambda_n} \right)\ge 1 - K_2(p\vee n)^{-1-c_2}.
	\end{equation*}
\end{enumerate}

Assumption \ref{a1} states that the predictors are uniformly bounded, which is reasonable as  predictors are often normalized during data pre-processing.
As defined  in \ref{a2}, $\Sigma = \diag (1, \Sigma_x)$, where $\Sigma_x$ is the variance-covariance matrix of $\bx$. The boundedness of the eigenvalues  of the variance-covariance matrix of  $\bx$ has been commonly assumed  in the  high dimensional literature \citep{zhao2006model,belloni2011l1,fan2014adaptive,van2014asymptotically}. \ref{a3} restricts the orders of $p$ and $n$  as well as the sparsity of $\bbeta^*$. Both \ref{a1} and \ref{a2}  guarantee the convergence of the {MLEs for the} low dimensional GLMs  (\ref{onetime})  with a  diverging number of predictors \citep{portnoy1985asymptotic,he2000parameters}. \ref{a3} requires $\mathcal{S}$  to possess the sure screening property, which  relaxes the  selection consistency assumption  in  \cite{fei2019drawing}.    

Variable selection methods that satisfy the sure screening property are available.
For example,  Assumptions \ref{a1} and \ref{a2}, along with a ``beta-min'' condition, which stipulates that $\min_{j\in S^*}|\beta^*_j|> c_0 n^{-\kappa}$ with $c_0>0,0<\kappa<1/2$,
ensure that  the commonly used  sure independence screening (SIS) procedure \citep{fan2010sure} satisfy 
the sure screening property; see Theorem 4 in \cite{fan2010sure}. {While a ``beta-min'' condition is common in deriving the sure screening property, it is not required for the de-biased type of estimators.}
We take $\mathcal{S}$ to be {the} SIS {procedure} when conducting variable selection in simulations and the data analysis.   
{Theorems  \ref{thm1} and  \ref{mainthm} correspond to the one-time estimator and the SSGLM estimator, respectively.}

	\begin{theorem}\label{thm1}
		Given model (\ref{glm1}) and assumptions \ref{a1}---\ref{a3}, consider the one-time estimator $\widetilde{\bbeta} = (\widetilde{\beta}_0,\widetilde{\beta}_1,\ldots,\widetilde{\beta}_p)^\rT$ as defined in (\ref{onetime}). Denote $p_s = |S|$ and $\widetilde{\sigma}_j^2 = \left(\{{I}^*_{S_{+j}} \}^{-1} \right)_{jj}, \;j\in \{0,1, \ldots,p\}$. Then as $n\rightarrow \infty$,
		\begin{itemize}
			\item [i.] $\| \widetilde{\bbeta}_{S_{+j} } - \bbeta^*_{S_{+j}} \|_2^2 = o_p(p_s/n)$, if $p_s\log p_s/n \rightarrow 0$;
			\item [ii.] $\sqrt{n_1}\left(\widetilde{\beta}_{j} - \beta^*_j\right)/\widetilde{\sigma}_{j} \stackrel{d}{\rightarrow} N(0,1)$, if  $p_s^2\log p_s/n \rightarrow 0.$
		\end{itemize}
	\end{theorem}
	
	\begin{theorem}\label{mainthm}
		Given model (\ref{glm1}) and under assumptions \ref{a1}---\ref{a3}
{and a partial orthogonality condition that $\{x_j, j \in S^*\}$ are independent of
$\{x_k, k \notin S^*\}$,} consider the smoothed estimator $\widehat{\bbeta}=(\widehat{\beta}_0,\widehat{\beta}_1,\ldots,\widehat{\beta}_p)^\rT$ as defined in (\ref{smest}). For each $j$,
define $\check{\sigma}_j^2 = \left( \{ {I}^*_{S^*_{+j}}\}^{-1}  \right)_{jj}$. Then, as $n,B\rightarrow \infty$,
		\begin{equation*}
		\sqrt{n}(\widehat{\beta}_j - \beta^*_j)/\check{\sigma}_j \stackrel{d}{\rightarrow} N(0, 1).
		\end{equation*}
	\end{theorem}
	
{The added partial orthogonality condition for Theorem \ref{mainthm} is a technical assumption for the validity of the theorem, which has been assumed in the high dimensional literature \citep{fan2008sure,fan2010sure,wang2014adaptive}. However, our numerical experiments suggest the robustness of our results to the violation of this condition.} In addition, while both of the one-time estimator $\widetilde{\beta}_j$ and the SSGLM estimator $\widehat{\beta}_j$ possess asymptotic consistency and normality, the key {advantage of $\widehat{\beta}_j$ over $\widehat{\beta}_j$ } {lies in} {the} efficiency. An immediate observation is that $\widehat{\beta}_j$ is estimated using all $n$ samples but $\widetilde{\beta}_j$ is estimated with only $n_1$ samples, which explains the different normalization constants in their respective variances, $\widehat{\sigma}^2_j/n$ and $\widetilde{\sigma}^2_j/n_1$. In addition,
with $\widetilde{\sigma}^2_j$  depending solely on a one-time variable selection $S$, its variability is high given the wide variability of $S$. On the other hand, $\widehat{\sigma}^2_j$ implicitly averages over the multiple  selections, $S^b$'s,  and gains  efficiency via  ``the effect of bagging'' \citep{buhlmann2002analyzing}; also see Web Table 1 of \cite{fei2019drawing} for empirical evidence under the linear regression setting. Moreover, the high variability of $\widetilde{\beta}_j$  may lead to a large false positive rate; see Figure 1 of \cite{fan2008sure}.


{We defer the proofs to the Appendix, but provide some intuition here. The randomness of the selection $\widehat{S}_{\lambda}$ presents  difficulties when developing the theoretical properties, but why sure screening works is that, given any subset $S \supseteq S^*$, the estimator $\widetilde{\bbeta}_S$ is consistent, though less efficient (with additional noise variables) than the ``oracle estimator'' $\widetilde{\bbeta}_{S^*}$ acting upon the true active set.  }
	The proof  also shows that $\widehat{\sigma}_j^2$ depends on the unknown $S^*$,  taking into account the variation in $B$ random splits. Therefore,  direct computation of $\widehat{\sigma}_j^2$ {in an} analytical form is not feasible.  Alternatively, we   estimate  the variance component via the infinitesimal jackknife method \citep{efron2014estimation,fei2019drawing}.
	
	\section{Variance Estimator and Inference by SSGLM}\label{s4}
	
	The infinitesimal jackknife method has been applied to estimate the variance of {the} bagged estimator with bootstrap-type resampling (sampling with replacement) \citep{efron2014estimation,fei2019drawing}.
	The idea is to treat each $\widetilde{\beta}^b_j$ as a function of the sub-sample $\bD_1^b$, or its	empirical distribution represented by the sampling indicator vector $\bJ_b = \left(J_{b1}, J_{b2}, \ldots, J_{bn} \right)$, where $J_{bi}\in\{0,1\}$ is {an} indicator {of} whether {the} $i^{th}$ observation is sampled in $\bD_1^b$. We further denote $J_{\cdot i} = \left(\sum_{b=1}^{B}J_{bi}\right)/B$. With slightly overused {notation,}  {let}
	\begin{gather*}
	    \widetilde{\beta}^b_j = t(\bD_1^b ) = t( \bJ_b;\bD^{(n)} );\\
	    \widehat{\beta}_j = \frac{1}{B}\sum_{b=1}^{B}\widetilde{\beta}^b_j  \stackrel{p}{\rightarrow} \bE^* t( \bJ_b;\bD^{(n)} ),\; \text{as } B \rightarrow \infty,
	\end{gather*}
	where {$t(\cdot)$ is a general function that maps the data to the estimator,} the expectation $\bE^*$ and the convergence are with respect to the probability measure induced by the randomness of $\bJ_b$'s. We can generalize the infinitesimal jackknife to estimate the variance, $\var\left(\widehat{\beta}_j \right)$, {analogous} to equation (8) of \cite{wager2018estimation}, as {follows}
	\begin{equation} \label{var1}
	\widehat{V}_j=\frac{n-1}{n} \left(\frac{n}{n-n_1}\right)^2 \sum_{i=1}^{n}\widehat{\mathrm{cov}}^2_{ij},
	\end{equation}
	where
	\begin{equation*}
	\widehat{\mathrm{cov}}_{ij}=\frac{1}{B} \sum_{b=1}^{B}\left( J_{bi}- J_{\cdot i}\right) \left( \widetilde{\beta}^b_j-\widehat{\beta}_j \right)   
	\end{equation*}
	is the covariance between the estimates $\widetilde{\beta}^b_j$'s and the sampling indicators $J_{bi}$'s with respect to the $B$ {splits}.
	Here, $n(n-1)/(n-n_1)^2$ is a finite-sample correction term with respect to the sub-sampling scheme. {Theorem 1 of} \cite{wager2018estimation} implies  that this variance estimator is consistent, in the sense that $\widehat{V}_j/\var\left(\widehat{\beta}_j \right)\xrightarrow{p} 1$  as $B\rightarrow \infty$.
	
	We further propose a  bias-corrected version of (\ref{var1}):
	\begin{gather} \label{var2}
	\widehat{V}_j^B = \widehat{V}_j - \frac{n}{B^2}\frac{n_1}{n-n_1}\sum_{b=1}^{B}(\widetilde{\beta}^b_j-\widehat{\beta}_j)^2.
	\end{gather}
	The derivation is similar to that in Section 4.1 of \cite{wager2014confidence}, but {it} is adapted to the sub-sampling scheme.
	The difference between $\widehat{V}_j$ and $\widehat{V}_j^B$ converges to zero with $n, B \rightarrow \infty$, as it can be re-written as
    $\frac{n}{B}\frac{n_1}{n - n_1} \widehat{v}_j,$
    {where} $\widehat{v}_j = B^{-1} \sum_{b=1}^{B}(\widetilde{\beta}^b_j-\widehat{\beta}_j)^2$ is the sample variance of  $\widetilde{\beta}^b_j$'s from $B$ {splits}. Thus both variance estimators are asymptotically equal. {See Algorithm \ref{alg2} for a summarized procedure of computing the bias-corrected variance estimates.}

{For finite samples, we give the order of $B$ to control the Monte Carlo errors of these two variance estimators.  First, with $n_1 = qn${ for} a fixed $0<q<1$,  the bias of $\widehat{V}_j$ is of order ${n\widehat{v}_j}/{B}$ \citep{wager2014confidence}.  Thus, setting $B = O(n^{1.5})$  will reduce the bias to the desired level of $O(n^{-0.5})$. On the other hand, $\widehat{V}_j^B$ effectively removes {this} bias, {as} it only requires $B = O(n)$ to control the Monte Carlo {Mean Squared Error (MSE)} to  $O(n^{-1})$  \citep{wager2014confidence}.  A comparison between $\widehat{V}_j$ and $\widehat{V}_j^B$, {given} in Simulation {\bf Example 1}, also shows the preference of  $\widehat{V}_j^B$ to $\widehat{V}_j$.}

 
	For $0 < \alpha <1$, the asymptotic $100(1-\alpha)\%$ confidence interval for  $\beta^*_j, j=1, \ldots, p,$  is given by
	\begin{equation*}
	\Big( \widehat{\beta}_j - \Phi^{-1}(1-\alpha/2) \sqrt{\widehat{V}_j^B} \; , \; \widehat{\beta}_j + \Phi^{-1}(1-\alpha/2) \sqrt{\widehat{V}_j^B} \Big),
	\end{equation*}
	{and} the p-value {for} testing $H_0: \beta^*_j = 0$ is
	\begin{equation*}\label{pv}
	2\times \Big\{1-\Phi\left({|\widehat{\beta}_j|/ \sqrt{\widehat{V}_j^B} }\right)\Big\},
	\end{equation*}
	where $\Phi$ is the CDF of the standard normal distribution.
	
	\section{Extension to Subvectors With Fixed Dimensions}\label{s5}

	We extend {the} SSGLM {procedure} to derive confidence regions for a subset of predictors and to test for {contrasts} of interest.
	Consider $\bbeta^*_{S^{(1)}}$ with $|S^{(1)}|=p_1\ge 2$, which is finite and does not increase with $n$ or $p$. Accordingly, the SSGLM estimator for it is presented in \textbf{Algorithm \ref{alg3}},
	and the extension of Theorem \ref{mainthm} is stated below.
	\begin{theorem}\label{thm3}
		Given model (\ref{glm1}) under assumptions \ref{a1}---\ref{a3} and a fixed finite subset $S^{(1)}\subset \{1,2,\ldots,p\}$ with $|S^{(1)}|=p_1$,  let  
		$\widehat{\bbeta}^{(1)}$ be the smoothed estimator for $\bbeta^*_{S^{(1)}}$ as defined in Algorithm \ref{alg3}. 
		Then as $n,B\rightarrow \infty$,
		\begin{equation*}
		\sqrt{n}  {I}^{(1)} \left(\widehat{\bbeta}^{(1)} - \bbeta^*_{S^{(1)}} \right)  \stackrel{d}{\rightarrow} N(0, \bI_{p_1}), 
		\end{equation*}
		where $\bI_{p_1}$ is a $p_1 \times p_1$  identity matrix, and  $	{I}^{(1)}$ is a $p_1 \times p_1$ positive definite matrix depending on $S^{(1)}$ and $S^*$ and is defined in the proof. 
	\end{theorem}
	
	There is a direct extension of the one-dimensional infinitesimal jackknife for estimating	the variance-covariance matrix of $\widehat{\bbeta}^{(1)}$,
	$\widehat{\Sigma}^{(1)} = {\widehat{\mathrm{COV}}_{(1)}}^\mathrm{T} \widehat{\mathrm{COV}}_{(1)}$,
	where
	\begin{gather*}
	\widehat{\mathrm{COV}}_{(1)} = \Big(\widehat{\mathrm{cov}}^{(1)}_1,\widehat{\mathrm{cov}}^{(1)}_2,\ldots,\widehat{\mathrm{cov}}^{(1)}_n\Big)^\mathrm{T},\;\mathrm{with} \\
	\widehat{\mathrm{cov}}^{(1)}_i = \sum_{b=1}^{B}(J_{bi}-J_{\cdot i})(\widehat{\bbeta}^b_{S^{(1)}}-\widehat{\bbeta}^{(1)})/B.
	\end{gather*}
	
	{To} test 
	$H_0: Q \bbeta^{(1)} = R,$  where $Q$ is an $r\times p_1$ {contrast} matrix and $R$ is an $r\times 1$ vector{, a} Wald-type test statistic can be formulated as
	\begin{equation}\label{test1}
	T = \left(Q\widehat{\bbeta}^{(1)} - R \right)^\rT\left[ Q\widehat{\Sigma}^{(1)} Q^\rT \right]^{-1}  \left(Q\widehat{\bbeta}^{(1)} - R \right),   
	\end{equation}
	which follows $\chi_r^2$ under $H_0$. Therefore, with a significance level  $\alpha \in (0,1)$, we reject $H_0$ when $T$ is larger than the $(1-\alpha)\times 100$ percentile of $\chi_r^2$.
	
\begin{algorithm}
	\caption{SSGLM for Subvector $\bbeta^{(1)}$\label{alg3}}
	\begin{algorithmic}[1]
		\REQUIRE A selection procedure $\mathcal{S}_\lambda$
		\INPUT Data $(\bY,\bX)$,   a data splitting proportion  ${q} \in (0,1)$,  the number of {splits} $B$, and an index set $S^{(1)}$ for the predictors of interest
		\OUTPUT Estimates of the coefficients of predictors indexed by $S^{(1)}$,    $\widehat{\bbeta}^{(1)}$
		\FOR{$b=1,2,\ldots,B$}
		Split the samples {into} two parts $\bD_1$ and $\bD_2$, with  $|\bD_1|={q} n$ and $|\bD_2|=(1- {q} )n$
		\STATE Apply $\mathcal{S}_\lambda$ {to} $\bD_2$ to select a subset of important predictors $S\subset \{1, \ldots,p\}$
		\STATE Fit a GLM by regressing $\bY^1$  on $\bX^1_{S^{(1)}\cup S}$, where $\bD_1 = (\bY^1,\bX^1)$ and compute the MLEs, denoted by $\widetilde{\bbeta}^{(1)}$
		\STATE Define $\widetilde{\bbeta}_{S^{(1)}}^b = \left(\widetilde{\bbeta}^{(1)} \right)_{S^{(1)}}$
		\ENDFOR
		\STATE Compute $\widehat{\bbeta}^{(1)}= \left(\sum_{b=1}^{B}\widetilde{\bbeta}^b_{S^{(1)}} \right)/B$
	\end{algorithmic}
\end{algorithm}

	\section{Simulations}\label{s6}
	
	We  compared the finite sample performance of the proposed SSGLM procedure, under various settings,  with two existing methods, the de-biased LASSO for GLMs \citep{van2014asymptotically,dezeure2015high} and the de-correlated score test \citep{ning2017general}, in  estimation accuracy and computation efficiency. We also investigated how  the choice of ${q} = n_1/n$, the splitting proportion, may impact the performance of SSGLM, explored various selection methods as part of the SSGLM procedure and their impacts on  estimation and inference,  illustrated our method with both logistic and Poisson regression {settings}, and assessed {the} power and type I error {of the procedure}.
	 We adopted some challenging simulation settings in \cite{buhlmann2014high}. For example, the indices of the active set, as well as  the non-zero effect sizes, were randomly {generated}, and various correlation structures were {explored}.

	\textbf{Example 1} investigated the performance of SSGLM with various splitting proportions and the convergence of the proposed variance estimators. We set $n_1 = {q} n,\;{q} =0.1,0.2,\ldots,0.9$. Under a  linear regression model, $Y_i = \bx_i\bbeta + \varepsilon_i,\ i=1,2,\ldots,n$ with i.i.d.  $\varepsilon_i \sim N(0,1)$, we set $n=500,\;p=1,000,\;s_0=10$ with an AR(1)  correlation structure, i.e. $\Sigma_{ij} = \rho^{|i-j|}, \rho = 0.5,\; i, j = 1,2,\ldots,p$. 
	The indices in the active set $S^*$ randomly varied from $\{1, \ldots,p\}$, and the non-zero effects of $\beta^*_j,j\in S^*$ were generated from $\mathrm{Unif}[(-1.5,-0.5)\cup(0.5,1.5)]$. For each ${q}$,  we computed the MSE for $\widehat{\beta}_j^{(k)}$, the smoothed estimate {of} $\beta_j$ from the $k$-th simulation, $k=1,2,\ldots,K$,
	\begin{gather*}
	\mathrm{MSE}_j = \frac{1}{K}\sum_{k=1}^{K}(\widehat{\beta}_j^{(k)} - \beta^*_j)^2,\quad
	\mathrm{MSE}_{\mathrm{avg}} = \frac{1}{p} \sum_{j=1}^{p} \mathrm{MSE}_j.
	\end{gather*}
	{The left} panel of Figure \ref{Fig1} showed that the minimum MSE was achieved when ${q} =0.5$, suggesting the rationality of equal-size splitting in practice. 
	
	However, the  MSE was, in general, less sensitive to  ${q}$ when $q$ was getting larger, hinting that a large $n_1$ may lead to adequate accuracy. Intuitively, there {is} a minimum sample size $n_2 = (1-q)n$ required for the selections to achieve the ``sure screening'' property{. For} example, LASSO with smaller sample size would select less variables given the same tuning parameter. On the other hand, larger $n_1 = qn$ improves the power of the low dimensional GLM estimators directly. Thus the optimal split proportion is achieved when $n_1$ is as large as possible, while  $n_2$ is large enough for the sure screening selection to hold. This intuition is also validated in Figure \ref{Fig1}, as efficiency is gained faster at the beginning due to better GLM estimators with larger $n_1${. This} gain is {then} out{weighed} by the bias due to poor selections with small $n_2$. Our conclusion is that {an} optimal split proportion exists, but depends on {the} specific selection method, the true model size, and {other factors, rather than being fixed.}
	
	{We further examined the convergence of the two variance estimators $\widehat{V}_j$ and $\widehat{V}^B_j$ proposed in (\ref{var1}) and (\ref{var2}) with respect to the number of splits, $B$. Under the same setting, and with $q = 0.5$, we calculated both $\widehat{V}_j$ and $\widehat{V}^B_j$ for  $B = 100,\ 200,\ \ldots,\ 2,000$, and compared these estimates with the empirical variance of $\widehat{\beta}_j$'s (considered to be the {\em truth}) based on $200$ simulation replicates. The right panel of Figure \ref{Fig1} plots the averages over all signals $j = 1,2,\ldots,p$ and shows  $\widehat{V}_j$ converges to the truth  much slower than $\widehat{V}^B_j$, and $\widehat{V}^B_j$ has small biases even with a relatively small $B$. }
	
	\textbf{Example 2} implemented  various selection methods,  LASSO, SCAD, MCP, Elastic net, and Bayesian LASSO,  when conducting variable selection for SSGLM,  and compared their impacts on estimation and inference.  Ten-fold cross-validation was used for {the} tuning parameters in each selection procedure.	We assumed a Poisson model with $n=300$, $p=400$, and $s_0=5$. 
	For $i=1, \ldots,\ n,$
		\begin{equation}\label{pois}
	\log\Big(\bE(Y_i|\bx_i) \Big) = \beta_0 + \bx_i\bbeta.
	\end{equation}
	Table \ref{tab1} {reports} the selection frequency for each $j$ out of $B$ splits. Larger $|\beta^*_j|$ yielded {a} higher selection frequency. For example, predictors with an absolute effect larger than 0.6 were selected  frequently. The average size of the selected models by each method varied from $23$ (for LASSO) to $8$ (for Bayesian LASSO). However, in terms of {the bias}, coverage probabilities, and mean squared errors, the {impact} of {the} different variable selection methods seemed to be negligible{. Thus,}  SSGLM was fairly robust to the {choice} of variable selection {method}.
	
	\textbf{Example 3} also assumed  model (\ref{pois}).
	We set $n=400,\ p=500$, {and} $s_0=6$, with non-zero coefficients between 0.5 and 1, and three  correlation structures: Identity; AR(1) with $\Sigma_{ij} = \rho^{|i-j|}, \rho = 0.5$; Compound Symmetry (CS)  with $\Sigma_{ij} = \rho^{I(i\ne j)}, \rho = 0.5$.
	
	Table \ref{sim_pois} {shows} that SSGLM {consistently} provided nearly unbiased estimates. The obtained standard errors (SEs) were close to the empirical standard deviations (SDs), leading to confidence intervals with coverage probabilities that were close to the 95\% nominal level.

	\textbf{Example 4} assumed a logistic regression  model for binary outcomes, with $n=400,\ p=500$, and $s_0 = 4$,
	\begin{equation}\label{logit}
	\logit\Big(\mathbf{P}(Y_i=1|\bx_i) \Big) = \beta_0 + \bx_i\bbeta.
	\end{equation}
	The index set for predictors with nonzero coefficients, $S^* = \{218,\ 242,\ 269,\ 417 \}$, were randomly generated, and $\bbeta^*_{S^*} = (-2,\ -1,\ 1,\ 2)$.
	We {report} the performance of SSGLM when inferring {the} subvector $ \bbeta^*_{S^*}$, in Tables \ref{sim_logistic2} and \ref{sim_logistic3}. Our method gave nearly unbiased estimates under different correlation structures and sufficient power for {the} various contrasts.

	\textbf{Example 5}
	compared our method with the de-biased LASSO estimator \citep{van2014asymptotically} and the de-correlated score test \citep{ning2017general} in {terms of} power and type I error.
	We assumed model (\ref{logit}) with $n=200,\;p=300$, $s_0=3$, {and} $\bbeta^*_{S^*} = (2,\ -2,\ 2)$ with AR(1) correlation structures.
	Table \ref{eg5} {summarises} the power of detecting each true signal and the average type I error for the noise variables under the AR(1) correlation structure with four correlation values,  $\rho = 0.25, 0.4, 0.6, 0.75$.
	
	Our method {was shown} to be the most powerful,  while maintaining the type I error around the nominal $0.05$ level. The power was over $0.9$ for the first three scenarios and was above $0.8$ with $\rho = 0.75$.  The de-biased LASSO estimators controlled the type I error well, but the power dropped from $0.9$ to {approximately} $0.67$ as the correlation among the predictors increased. The de-correlated score tests had the least power and the highest type I error.
	{While these two competing methods have the same efficiency asymptotically, they do differ by specific implementations, for example, the choice of tuning parameters. Indeed, de-biased methods may be sensitive to tuning parameters, which could explain the gap in the finite sample performance. }
	
	Table \ref{eg5} {summarizes} the average computing time (in seconds) of the three methods per data set (R-3.6.2 on {an} 8-core MacBook Pro). On average, our method took  $17.7$ seconds, which was {the} fastest among the three methods. The other two methods were slower for the smaller $\rho$'s (75 and 37 seconds, respectively) and faster for the larger $\rho$'s (41 and 18 seconds, respectively), likely because the node-wise LASSO procedure that was used for estimating the precision matrix tended to be faster when handling more highly correlated predictors.
	
	
	
	\section{Data Example}\label{s7}
	
	
	We analyzed a subset of the BLCSC data \citep{Christiani2017blcs}, consisting of $n = 1,459$ individuals, among {whom} $708$ {were} lung cancer patients and $751$ {were} controls. {After cleaning, the data contained}  $6,829$ SNPs, along with important demographic variables including age, gender, race, education level, and smoking status (Table \ref{blcs_dem}). {As} smoking {is known to play} a {significant} role in {the development of} lung cancer, we were particularly interested in estimating the interactions between the SNPs and smoking status, in addition to {their} main effects.
	
	We assumed a high-dimensional logistic model with the {binary} outcome being {an indicator of} lung cancer {status. Predictors included} demographic variables, the SNPs (with prefix ``AX"), and the interactions between the SNPs and smoking status
	({with} prefix ``SAX"; $p=13,663$). 
	We applied {the} SSGLM with $B=1,000$ random {splits} and drew inference on all 13,663 predictors. Table \ref{blcs} {lists} the top predictors ranked by their p-values. We identified $9$ significant coefficients after using Bonferroni correction {for multiple comparisons}. All were  interaction terms, providing strong evidence of  SNP-smoking interactions, {which have} rarely {been} reported.
	These nine SNPs came from three genes, TUBB, ERBB2, and TYMS.  TUBB mutations {are} associated with both poor treatment response to paclitaxel-containing chemotherapy and poor survival in patients with advanced non-small-cell lung cancer (NSCLC) \citep{monzo1999paclitaxel,kelley2001genetic}. \cite{rosell2001predictive} has proposed using the presence of TUBB mutations as a basis for selecting initial chemotherapy for patients with advanced NSCLC. In contrast, intragenic ERBB2 kinase mutations occur more often in the adenocarcinoma  lung cancer {subtype} \citep{stephens2004lung,beer2002gene}.
	{Lastly}, advanced NSCLC patients with low/negative thymidylate synthase (TYMS) {are shown to} have better {responses} to Pemetrexed–{based chemotherapy} and longer progression free survival \citep{wang2013association}.
	
	For {comparisons}, we applied the de-sparsified estimator for GLM \citep{buhlmann2014high}. A direct application of the ``lasso.proj" function in the ``hdi" R package \citep{dezeure2015high} was not feasible given the data size. Instead, we used a shorter sequence of  candidate $\lambda$ values and $5$-fold instead of $10$-fold cross validation for the node-wise LASSO  procedure{. This procedure} costs {approximately} one day of CPU time. After correcting for multiple testing, there were two significant coefficients, both of which were interaction terms corresponding to SNPs AX.35719413\_C and AX.83477746\_A. Both SNPs were from the TUBB gene, and the first SNP was also identified by our method.

	To validate our findings, we applied the prediction accuracy measures for nonlinear models proposed in \cite{li2019prediction}. We calculated the $R^2$, the proportion of variation {explained} in $\bY$, for the models we chose to compare. We {report five} models {and their corresponding} $R^2$ {values}: \textbf{Model 1.{~}}the baseline model including only the demographic variables {($R^2 = 0.0938$)}; \textbf{Model 2.{~}}the baseline model plus the significant interactions after {the} Bonferroni correction {in} Table \ref{blcs} {($R^2 = 0.1168$)}; 
	\textbf{Model 3.{~}}Model 2 plus the main effects of its interaction terms {($R^2 = 0.1181$)}; 
	\textbf{Model 4.{~}}the baseline model plus the significant interactions from the de-sparsified LASSO method {($R^2 = 0.1018$)}; 
	\textbf{Model 5.{~}}Model 4 plus the corresponding main effects {($R^2 = 0.1076$)}. 
	Model 2 based on our method {explained} $25\%$ more variation in $\bY$ {than the baseline model} (from $0.0938$ to $0.1168$), while Model 4 based on the de-sparsified LASSO method only explains $8.5\%$ more variation (from $0.0938$ to $0.1018$). We also plotted {Receiver-Operating Characteristic} (ROC) curves {for} models 1, 2, and 4 (Figure \ref{roc}){. Their corresponding areas under the curves} (AUCs) were $0.645,\ 0.69,\text{ and } 0.668$, respectively.

	{Previous} literature has identified several SNPs as potential risk factors for lung cancer. We studied a controversial SNP, rs3117582, from the TUBB gene on chromosome 6. This SNP was identified in association with lung cancer risk in a case/control study by \cite{wang2008common}, while on the other hand,  \cite{wang2009role} found no evidence of association between the SNP and risk of lung cancer among \emph{never-smokers}. 
	Our goal was to test this SNP and its interaction with smoking in the presence of all the other predictors under the high dimensional logistic model.
Slightly {overusing} notation, we denoted the coefficients corresponding to rs3117582 and its interaction with smoking as $\bbeta^{(1)} = (\beta_1,\beta_2)$, and tested $H_0: \beta_1 = \beta_2 = 0$. Applying the proposed method, we obtained
	\begin{equation*}
	(\widehat{\beta}_1,\widehat{\beta}_2) = (-0.067,0.005),\;\widehat{\mathrm{COV}}\left(\widehat{\beta}_1,\widehat{\beta}_2\right) = \begin{pmatrix}
	0.44,\; -0.43\\
	-0.43,\; 0.50
	\end{pmatrix}.
	\end{equation*}
	The test statistic of the overall effect was $T = 0.062$ by (\ref{test1}) with a p-value of $0.97$, which concluded that, among the patients in  {BLCSC},  rs3117582 was not significantly related to lung cancer, regardless of the smoking status.

	\section{Conclusions}\label{s8}

		Our approach {for drawing} inference, by  adopting {a} ``split and smoothing'' idea, {improves} upon \cite{fei2019drawing} which used bootstrap resampling, and 
		 {recasts} a high dimensional inference problem into a sequence of low dimensional estimations. Unlike many of the existing methods \citep{zhang2014confidence,buhlmann2014high,javanmard2018debiasing}, our method is more computationally feasible as it does not require estimating  high dimensional precision matrices.
		{Our algorithm enables us to make full use of parallel computing for improved computational efficiency, because fitting the $p$ low dimensional GLMs and  randomly splitting the data $B$ times are both separable tasks, which can be implemented in parallel.}
		
		We have derived the variance estimator using the infinitesimal jackknife method adapted to the splitting and smoothing procedure \citep{efron2014estimation,wager2018estimation}{. This estimator} is free of parametric assumptions, resembles  bagging \citep{buhlmann2002analyzing}, and leads to confidence intervals with correct coverage probabilities. Moreover, we have relaxed the stringent {\it selection consistency} assumption on variable selection  as required in  \cite{fei2019drawing}. We have shown that our procedure works with a  mild {\it sure screening} assumption for the selection method. 
		
		
		There are open problems to be addressed. First, our method relies on {a} sparsity condition {for the} model parameters. We envision that relaxation of the condition may take a major effort, though our preliminary simulations (Example B.2 in Appendix B)  suggest that our procedure might work when the sparsity condition fails. Second, as our model is fully parametric, in-depth research is needed to develop {a} more robust {approach} when the model is mis-specified. Finally, while our procedure is  feasible when $p$ is large (tens of thousands), the computational  cost  increases substantially when $p$ is extraordinarily large (millions). Much effort is warranted to  enhance {its} computational efficiency. {Nevertheless, our work does provide a starting point for future investigations.}
		
		\section*{Acknowledgements}
		We are grateful towards Dr.~Boaz Nadler and  three referees for the insightful comments that have helped improve the manuscript. We thank 
		Stephen Salerno, Department of Biostatistics, University of Michigan, for proofreading the manuscript and for the edits that have bettered the presentation of the manuscript. We thank our long time collaborator, Dr.~David Christiani, Harvard Medical School, for providing the BLCSC data. The work is supported by  grants from NIH
		(R01CA249096,R01AG056764 and U01CA209414).

	

\bibliography{mybib}
	
\newpage
\begin{figure}
	\caption{Left: Average MSEs of all predictors at split proportions ${q}$'s from 0.1 to 0.9. Right: Convergence of two variance estimators as $B$ increases. \label{Fig1}}
	\centering
	\includegraphics[scale=0.42]{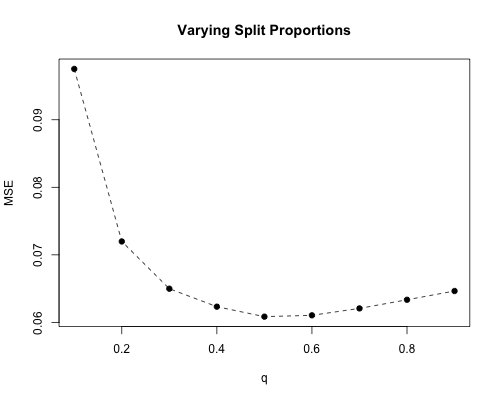}
	\includegraphics[scale=0.42]{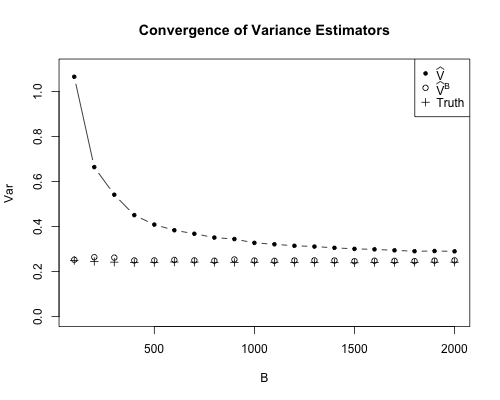}
\end{figure}

\begin{figure}
	\centering
	\caption{ROC curves of the three selected models.\label{roc}}
	\includegraphics[scale=0.75]{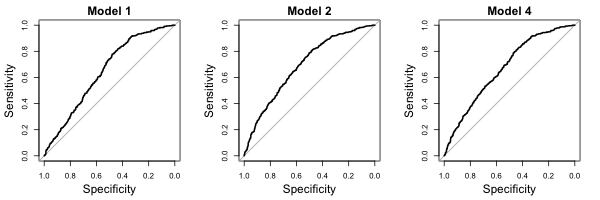}
\end{figure}

\begin{table}[ht]
	\centering
	\caption{Comparisons of different selection procedures to implement our proposed method. The first column is the indices of the non-zero signals. The last row for the selection frequency is the average number of predictors being selected by each procedure. The last row for the coverage probability is the average coverage probability of all predictors.\label{tab1}}
	\begin{tabular}{rrrrrrr}
		\hline
		Index $j$ & $\beta^*_j$ & LASSO & SCAD & MCP & EN & Bayesian  \\ 
		\hline
		\multicolumn{7}{c}{Selection frequency } \\ 
		12 & 0.4 & 0.59 & 0.55 & 0.49 & 0.60 & 0.60   \\ 
		71 & 0.6 & 0.93 & 0.92 & 0.90 & 0.95 & 0.94   \\ 
		351 & 0.8 & 0.99 & 0.99 & 0.99 & 1.00 & 1.00   \\ 
		377 & 1.0 & 1.00 & 1.00 & 1.00 & 1.00 & 1.00   \\ 
		386 & 1.2 & 1.00 & 1.00 & 1.00 & 1.00 & 1.00   \\ 
		Average model size && 23.12 & 13.15 & 10.89 & 10.31 & 7.98  \\ 
		\hline
		\multicolumn{7}{c}{Bias}   \\ 
		12 & 0.4 & 0.003 & 0.003 & 0.003 & 0.003 & 0.001 \\ 
		71 & 0.6 & 0.007 & 0.008 & 0.008 & 0.008 & -0.010 \\ 
		351 & 0.8 & -0.001 & 0.001 & 0 & 0 & 0.001 \\ 
		377 & 1.0 & -0.005 & -0.005 & -0.006 & -0.005 & 0.001 \\ 
		386 & 1.2 & 0.002 & 0.001 & 0.001 & 0.001 & 0.004 \\ 
		\hline
		\multicolumn{7}{c}{Coverage probability } \\ 
		12 && 0.90 & 0.90 & 0.91 & 0.91 & 0.95 \\ 
		71 && 0.94 & 0.94 & 0.95 & 0.94 & 0.94 \\ 
		351 && 0.95 & 0.95 & 0.95 & 0.94 & 0.95 \\ 
		377 && 0.94 & 0.93 & 0.93 & 0.94 & 0.92 \\ 
		386 && 0.94 & 0.95 & 0.95 & 0.95 & 0.94 \\ 
		Average && 0.93 & 0.94 & 0.94 & 0.94 & 0.94 \\ 
		\hline
		\multicolumn{7}{c}{MSE } \\ 
		12 && 0.111 & 0.110 & 0.110 & 0.109 & 0.106 \\ 
		71 && 0.104 & 0.103 & 0.102 & 0.102 & 0.101 \\ 
		351 && 0.103 & 0.103 & 0.103 & 0.103 & 0.100 \\ 
		377 && 0.101 & 0.100 & 0.100 & 0.100 & 0.109 \\ 
		386 && 0.097 & 0.096 & 0.096 & 0.096 & 0.102 \\ 
		Average && 0.105 & 0.104 & 0.103 & 0.103 & 0.102 \\ 
		\hline
	\end{tabular}
\end{table}

\begin{table}[ht]
	\centering
	\caption{SSGLM under the Poisson regression and three correlation structures. Bias, average standard error (SE), empirical standard deviation (SD), coverage probability (Cov prob), and selection frequency (Sel freq) are reported. The last column summarizes the average of all noise variables. \label{sim_pois}}
	\begin{tabular}{rrrrrrrrrr}
		\hline
		&Index $j$ & 0 (Int) & 74 & 109 & 347 & 358 & 379 & 438 & - \\ 
		&$\beta^*_j$ & 1.000 & 0.810 & 0.595 & 0.545 & 0.560 & 0.665 & 0.985 & 0 \\ 
		\hline
		Identity& Bias & -0.010 & 0 & 0 & 0.001 & 0.005 & 0.005 & 0.006 & 0 \\ 
		&SE & 0.050 & 0.035 & 0.034 & 0.035 & 0.035 & 0.034 & 0.035 & 0.034 \\ 
		&SD & 0.064 & 0.036 & 0.038 & 0.031 & 0.033 & 0.038 & 0.036 & 0.036 \\ 
		&Cov prob & 0.870 & 0.920 & 0.900 & 0.960 & 0.990 & 0.910 & 0.950 & 0.936 \\ 
		&Sel freq & 1.000 & 1.000 & 1.000 & 1.000 & 1.000 & 1.000 & 1.000 & 0.015 \\ 
		\hline 
		AR(1)	& Bias & 0.006 & 0.003 & -0.002 & -0.001 & -0.001 & -0.005 & 0.003 & 0 \\ 
		&SE & 0.052 & 0.035 & 0.035 & 0.035 & 0.035 & 0.035 & 0.035 & 0.035 \\ 
		&SD & 0.056 & 0.031 & 0.041 & 0.035 & 0.037 & 0.037 & 0.037 & 0.036 \\ 
		&Cov prob & 0.930 & 0.970 & 0.890 & 0.960 & 0.950 & 0.930 & 0.960 & 0.937 \\ 
		&Sel freq & 1.000 & 1.000 & 1.000 & 1.000 & 1.000 & 1.000 & 1.000 & 0.015 \\ 
		\hline
		CS	& Bias & -0.003 & -0.005 & 0.004 & -0.002 & 0.005 & -0.004 & -0.001 & 0.001 \\ 
		&SE & 0.033 & 0.043 & 0.043 & 0.042 & 0.043 & 0.043 & 0.044 & 0.042 \\ 
		&SD & 0.038 & 0.046 & 0.044 & 0.052 & 0.040 & 0.047 & 0.043 & 0.044 \\ 
		&Cov prob & 0.960 & 0.900 & 0.930 & 0.900 & 0.970 & 0.910 & 0.950 & 0.934 \\ 
		&Sel freq & 1.000 & 1.000 & 0.999 & 0.997 & 0.998 & 0.999 & 1.000 & 0.016 \\ 
		\hline
	\end{tabular}
\end{table}

\begin{table}[ht]
	\centering
	\caption{SSGLM under the logistic regression, with estimation and inference for the subvector $\bbeta^{(1)} = \bbeta_{S^*}$. We compare the SSGLM (left) to the oracle model (right), where the oracle estimator is from the low dimensional GLM given the true set $S^*$, and the empirical covariance matrix is with respect to the simulation replications. \label{sim_logistic2}}
	\begin{tabular}{rrrrr rrrrr}
		\toprule
		Index $j$ &  218 & 242 & 269 & 417 & Index $j$ & 218 & 242 & 269 & 417\\ 
		$\beta^*_j$ &  -2 & -1 & 1 & 2 & $\beta^*_j$ &  -2 & -1 & 1 & 2\\ 
		\hline
		\multicolumn{10}{c}{Identity} \\
		$\widehat{\bbeta}^{(1)}$ & -2.048 & -1.043 & 0.999 & 2.096 &Oracle & -1.995 & -1.026 & 0.973 & 2.043 \\  
		$\widehat{\Sigma}^{(1)}$ &0.146 & 0.010 & -0.009 & -0.020 &Empirical & 0.155 & 0.006 & -0.009 & -0.027 \\ 
		& 0.010 & 0.134 & -0.004 & -0.011 && 0.006 & 0.129 & -0.011 & -0.015 \\ 
		& -0.009 & -0.004 & 0.134 & 0.009 && -0.009 & -0.011 & 0.152 & 0.010 \\ 
		& -0.020 & -0.011 & 0.009 & 0.143 && -0.027 & -0.015 & 0.010 & 0.134 \\ 
		\hline
		\multicolumn{10}{c}{AR(1)} \\
		$\widehat{\bbeta}^{(1)}$  & -2.073 & -1.014 & 1.002 & 2.110 &Oracle & -2.024 & -0.991 & 0.977 & 2.062 \\ 
		$\widehat{\Sigma}^{(1)}$ & 0.145 & 0.012 & -0.011 & -0.023 	&Empirical &0.141 & 0.012 & -0.016 & -0.028 \\ 
		& 0.012 & 0.137 & -0.006 & -0.011 && 0.012 & 0.112 & -0.006 & 0 \\ 
		& -0.011 & -0.006 & 0.135 & 0.010  && -0.016 & -0.006 & 0.129 & 0.009 \\ 
		& -0.023 & -0.011 & 0.010 & 0.147  && -0.028 & 0 & 0.009 & 0.136 \\ 
		\hline
		\multicolumn{10}{c}{CS} \\
		$\widehat{\bbeta}^{(1)}$  & -2.095 & -1.033 & 1.070 & 2.102 &Oracle & -2.037 & -1.024 & 1.027 & 2.028 \\ 
		$\widehat{\Sigma}^{(1)}$ &  0.223 & -0.026 & -0.048 & -0.063 &Empirical & 0.192 & -0.030 & -0.044 & -0.045 \\  
		& -0.026 & 0.208 & -0.043 & -0.047 && -0.030 & 0.187 & -0.037 & -0.044 \\ 
		& -0.048 & -0.043 & 0.207 & -0.028 && -0.044 & -0.037 & 0.165 & -0.011 \\ 
		& -0.063 & -0.047 & -0.028 & 0.224 && -0.045 & -0.044 & -0.011 & 0.179 \\ 
		\hline
	\end{tabular}
\end{table}

\begin{table}
	\centering
	\caption{SSGLM under Logistic regression, with rejection rates of testing the contrasts. When the truth is $0$, the rejection rates estimate the type I error probability; when the truth is nonzero, they estimating the testing power. 
		\label{sim_logistic3}}
	\begin{tabular}{rrrrr}
		\toprule
		$H_0$ & Truth & Identity & AR(1) & CS\\
		\hline
		$\beta^*_{218} + \beta^*_{417} = 0$ & 0 & 0.05 & 0.04 & 0.03\\
		$\beta^*_{242} + \beta^*_{269} = 0$ & 0 & 0.06 & 0.04 & 0.025 \\
		$\beta^*_{218} + \beta^*_{269} = 0$ & $-1$ & 0.56 & 0.57 & 0.42\\
		$\beta^*_{242} + \beta^*_{417} = 0$ & 1 & 0.55  & 0.58 & 0.48\\
		$\beta^*_{242} = 0$ & $-1$ & 0.83 & 0.80 & 0.61\\
		$\beta^*_{269} = 0$ & 1 & 0.74 & 0.81 & 0.70\\
		$\beta^*_{218} = 0$ & $-2$ & 1 & 1 & 1\\
		$\beta^*_{417} = 0$ & 2 & 1 & 1 & 1\\
		\hline
	\end{tabular}
\end{table}

\begin{table}[ht]
	\centering
	\caption{Comparisons of SSGLM, Lasso-pro, and De-correlated score test (Dscore) in  power, type I error and computing time. AR(1) correlation structure with different $\rho$'s for $\bX$ are assumed. \label{eg5}}
	\begin{tabular}{rrrrr r}
		\hline
		& \multicolumn{3}{c}{Power} & Type I error & Time \\
		\hline
		Truth & $\beta^*_{10} = 2$ & $\beta^*_{20} = -2$ & $\beta^*_{30} = 2$ & $\beta^*_j = 0$ & (secs) \\ 
		\hline
		$\rho=0.25$ Proposed & 0.920 & 0.930 & 0.950 & 0.049 & 17.7\\ 
		Lasso-pro & 0.900 & 0.930 & 0.900 & 0.042 & 74.7\\ 
		Dscore	& 0.790 & 0.880 & 0.890 & 0.177 & 37.0\\ 
		\hline
		$\rho=0.4$ Proposed & 0.940 & 0.960 & 0.965 & 0.049 & 17.6\\ 
		Lasso-pro & 0.920 & 0.910 & 0.920 & 0.043  & 66.0\\ 
		Dscore	& 0.770 & 0.905 & 0.840 & 0.175  & 30.7\\  
		\hline
		$\rho=0.6$ Proposed  & 0.940 & 0.950 & 0.880 & 0.054 & 17.7\\ 
		Lasso-pro & 0.850 & 0.750 & 0.850 & 0.045  & 53.3\\ 
		Dscore	& 0.711 & 0.881 & 0.647 & 0.268 & 20.1\\ 
		\hline
		$\rho=0.75$ Proposed & 0.863 & 0.847 & 0.923 & 0.060  & 17.7\\ 
		Lasso-pro & 0.690 & 0.670 & 0.650 & 0.053 & 41.0\\ 
		Dscore	 & 0.438 & 0.843 & 0.530 & 0.400 & 17.9\\ 
		\hline
	\end{tabular}
\end{table}

\begin{table}[ht]
	\centering
	\caption{Demographic characteristics of the BLCSC SNP data.\label{blcs_dem}}
	\begin{tabular}{rrr}
		\hline
		& Controls ($751$) & Cases ($708$) \\ 
		\hline
		Race&&\\
		White & 726 & 668 \\ 
		Black &   5 &  22 \\ 
		Other &  20 &  18 \\ 
		\hline
		Education &&\\
		$<$High school &  64 &  97 \\ 
		High school & 211 & 181 \\ 
		$>$High school & 476 & 430 \\ 
		\hline
		Age &&\\
		Mean(sd) & 59.7(10.6)  & 60(10.8)\\
		\hline
		Gender && \\ 
		Female & 460 & 437 \\ 
		Male & 291 & 271 \\ 
		\hline
		Pack years &&\\
		Mean(sd) & 18.8(25.1)  & 46.1(38.4)\\
		\hline
		Smoking &&\\
		Ever & 498 & 643 \\ 
		Never & 253 &  65 \\ 
		\hline
	\end{tabular}
\end{table}

\begin{table}[ht]
	\caption{SSGLM fitted to the BLCSC data. SNP variables start with ``AX"; interaction terms start with ``SAX"; ``Smoke" is  binary (1=ever smoked, 0=never smoked). Rows are sorted by p-values. \label{blcs}}
	\centering
	\begin{tabular}{rrrrrrr}
		\toprule
		Variable & $\widehat{\beta}$ & SE & T & p-value & Adjusted P & Sel freq \\ 
		\hline
		SAX.88887606\_T & 0.33 & 0.02 & 17.47 & $<10^{-3}$ & $<0.01$ & 0.08 \\ 
		SAX.11279606\_T & 0.53 & 0.06 & 8.23 & $<10^{-3}$ & $<0.01$ & 0.00 \\ 
		SAX.88887607\_T & 0.29 & 0.04 & 6.97 & $<10^{-3}$ & $<0.01$ & 0.01 \\ 
		SAX.15352688\_C & 0.56 & 0.08 & 6.90 & $<10^{-3}$ & $<0.01$ & 0.01 \\ 
		SAX.88900908\_T & 0.54 & 0.09 & 5.95 & $<10^{-3}$ & $<0.01$ & 0.02 \\ 
		SAX.88900909\_T & 0.51 & 0.09 & 5.69 & $<10^{-3}$ & $<0.01$ & 0.02 \\ 
		SAX.32543135\_C & 0.78 & 0.14 & 5.49 & $<10^{-3}$ & $<0.01$ & 0.25 \\ 
		SAX.11422900\_A & 0.32 & 0.06 & 5.24 & $<10^{-3}$ & $<0.01$ & 0.09 \\ 
		SAX.35719413\_C & 0.47 & 0.10 & 4.63 & $<10^{-3}$ & 0.049 & 0.00 \\ 
		\hline 
		SAX.88894133\_C & 0.43 & 0.10 & 4.53 & $<10^{-3}$ & 0.08 & 0.00 \\ 
		SAX.11321564\_T & 0.47 & 0.11 & 4.44 & $<10^{-3}$ & 0.12 & 0.00 \\ 
		\ldots&&&&&& \\ 
		AX.88900908\_T & 0.40 & 0.11 & 3.84 & $<10^{-3}$ & 1.00 & 0.00 \\ 
		Smoke & 0.89 & 0.23 & 3.82 & $<10^{-3}$ & 1.00 & - \\ 
		\ldots&&&&&& \\ 
		\hline
	\end{tabular}
\end{table}

\clearpage

\newpage
\section*{Appendix A: Proofs of Theorems}

\begin{proof}of Theorem \ref{thm1}:
	
	From the data split, $\bD_1$ and $\bD_2$ are mutually exclusive, thus $S$, from $\bD_2$, is independent of $\bD_1 = (\bY^1,\bX^1)$. We show the asymptotics of $\widetilde{\bbeta}_{S_{+j} }$ in (\ref{onetime}) with a diverging number of parameters $p_s$, by using the techniques and results from \cite{he2000parameters,niemiro1992asymptotics}.
	Without loss of generality, and to simplify notation, we let $j=1\in S$. Then $S_{+j} = S$. The argument is the same	if $1\notin S$ and for any other $j$. 
	
	To proceed, we first restrict on the event of
	$\Omega= \{ S \supseteq S^*\}$. With Assumption (A3), $\pr(\Omega) \ge 1 - K_2(p\vee n_2)^{-1-c_2}.$
	Recall that
	\begin{gather*}
	\widetilde{\bbeta}_{S_{+j} } = 
	\argmin_{\bbeta \in \bR^{|S|+1} } \ell_{S} (\bbeta_{S}) = \argmin_{\bbeta \in \bR^{|S|+1} } \ell (\bbeta_{S};\bY^1,\bX^1_{S});\\
	\widetilde{\beta}_{1} = \left(\widetilde{\bbeta}_{S_{+j} } \right)_1.   
	\end{gather*}
	To apply Theorems 2.1 and 2.2 of \cite{he2000parameters} in the GLM case, we can verify that our Assumptions  \ref{a1} and \ref{a2} will lead to their conditions (C1), (C2), (C4) and (C5)  with $C=1, r=2$ and $A(n,p_s) = p_s$. To verify their (C3), we first note that their $D_n$ is our ${I}^*_{S}$.
	Then for any $\bbeta_{S}$, $ \alpha \in \bR^{p_s}$ such that $\|\alpha\|_2 = 1$, a second order Taylor expansion of $U_{S}(\bbeta_{S})$ around  $\bbeta^*_{S}$ leads to  
	\[
	\left| \alpha^\rT\bE_{\bbeta^*} \left( U_{S}(\bbeta_{S})  - U_{S}(\bbeta^*_{S})  \right) - \alpha^\rT{I}^*_{S} \left( \bbeta_{S} - \bbeta^*_{S}\right) \right| \le  O(\| \bbeta_{S} - \bbeta^*_{S}\|_2^2).
	\]
	Hence,
	\[\sup_{ \| \bbeta_{S} - \bbeta^*_{S}\|\le (p_s/n)^{1/2}} \left| \alpha^\rT\bE_{\bbeta^*} \left( U_{S}(\bbeta_{S})  - U_{S}(\bbeta^*_{S})  \right) - \alpha^\rT{I}^*_{S} \left( \bbeta_{S} - \bbeta^*_{S}\right) \right|\\
	 \le O(p_s/n) = o(n^{1/2}),
	\]
	which means their (C3) follows. Thus, by Theorem 2.1 of \cite{he2000parameters}, 
	\[\| \widetilde{\bbeta}_{S } - \bbeta^*_{S} \|_2^2 = o_p(p_s/n_1),
	\]given $p_s\log p_s/n_1 \rightarrow 0$.
		Furthermore, by Theorem 2.2 of \cite{he2000parameters}, if $p_s^2\log p_s/n_1 \rightarrow 0$, 
	then 
	\[\|n_1^{1/2} (\widetilde{\bbeta}_{S} -\bbeta^*_{S} ) +  n_1^{-1/2} \{{I}^*_{S} \}^{-1} U_{S}(\bbeta^*_{S})\|_2 = o_p(1).
	\]
	
Releasing the restriction on $\Omega$ and with $ \pr(\Omega^c) =\pr ( S \not\supseteq S^*) \le K_2(p\vee n_2)^{-1-c_2}$,  we would still have
$\| \widetilde{\bbeta}_{S } - \bbeta^*_{S} \|_2^2 =  o_p(p_s/n_1) $, given $p_s\log p_s/n_1 \rightarrow 0$.
To see this, for any $\epsilon>0$, we can consider
\begin{eqnarray*} 
&  & \pr ( \|(n_1/p_s)^{1/2} (\widetilde{\bbeta}_{S } - \bbeta^*_{S}) \|_2> \epsilon)  \\
& < & \pr ( \|(n_1/p_s)^{1/2} (\widetilde{\bbeta}_{S } - \bbeta^*_{S}) \|_2> \epsilon| \Omega) \pr(\Omega)
+ \pr(\Omega^c) \\
& < & \pr ( \|(n_1/p_s)^{1/2} (\widetilde{\bbeta}_{S } - \bbeta^*_{S}) \|_2> \epsilon| \Omega) + K_2(p\vee n_2)^{-1-c_2},
\end{eqnarray*}
where both terms in the last inequality converge to 0
as $n_1 \rightarrow \infty$ and $n_2=(1-q)n_1/q$, with $0<q<1$ a constant. Similarly, we can show
	$\|n_1^{1/2} (\widetilde{\bbeta}_{S} -\bbeta^*_{S} ) +  n_1^{-1/2} \{{I}^*_{S} \}^{-1} U_{S}(\bbeta^*_{S})\|_2 = o_p(1),$ if $p_s^2\log p_s/n_1 \rightarrow 0$, which can also be written as 
	\begin{equation} \label{junk2}
	\widetilde{\bbeta}_{S} -\bbeta^*_{S}  = - n_1^{-1} \{{I}^*_{S} \}^{-1} U_{S}(\bbeta^*_{S}) + r_{n_1},
	\end{equation}
	with $\|r_{n_1}\|_2^2 = o_p(1/n_1)$.
	Consequently, by taking $\alpha = (0,1,0,\ldots,0)^\rT$  and left-multiplying  both sides of (\ref{junk2}) by
	$n^{1/2} \alpha^\rT $, we have 
	\[\sqrt{n_1}\left(\widetilde{\beta}_{1} - \beta^*_1\right)/\widetilde{\sigma}_{1} \stackrel{d}{\rightarrow} N(0,1),\] where $\widetilde{\sigma}_1^2 = \left( \{{I}^*_{S}\}^{-1}  \right)_{11}.$ 
	
	
\end{proof}

The following lemma, which is needed for the proof of Theorem \ref{mainthm},  bounds the estimates of coefficients, when the selected subset $S^b$ misses important predictors in $S^*$ for some $ 1 \le b \le B$. Although $S^* \not\subseteq S^b$  with probability going to zero by Assumption \ref{a3}, we  need to establish an upper bound in order to control the bias of  $\widehat{\beta}_j$ for any $j$. 
\begin{lemma} \label{l1}
	With model (\ref{glm1}) and Assumptions \ref{a1} and \ref{a2}, consider the GLM estimator $\widetilde{\bbeta}_{S }$ with respect to subset $S$ as defined in (\ref{onetime}). Denote by $p_s = |S|$.  If $p_s\log p_s/n \rightarrow 0$, then with probability going to 1,
	$
	|\widetilde{\bbeta}_{S }|_\infty \le C_\beta,$
	where $C_\beta>0$ is a constant depending on $c_{\min}, c_{\max}, c_\beta$, and $A(0)$.
\end{lemma}
\begin{proof} 
By definition,
	\begin{gather*}
	\widetilde{\bbeta}_{S_{+j} } = 
	\argmin_{\bbeta_{S}\in \bR^{p_s+1}} \ell_{S} (\bbeta_{S}) = \argmin_{\bbeta_{S}\in \bR^{p_s+1}} \ell (\bbeta_{S};\bY^1,\bX^1_{S}).
	\end{gather*}	
	If $S^* \subseteq S$, the result immediately follows from Theorem \ref{thm1} by taking $C_\beta = 2c_\beta$.
	When $S^* \not\subseteq S$, the minimizer $\widetilde{\bbeta}_{S_{+j} }$ is not an unbiased estimator of $\bbeta^*_{S}$ anymore. However, we show that the boundedness of $\widetilde{\bbeta}_{S_{+j} }$ is guaranteed from the strong convexity of the objective function $\ell_{S} (\bbeta_{S})$.
	
	To see this, we note that the observed information is $\nabla^2 \ell_S (\bbeta_S)= \widehat{I}_S(\bbeta_S) = \frac{1}{n} {\overline{\bX}_S}^\rT \bV_S{\overline{\bX}_S},$
	where $\bV_S = \mathrm{diag}\{ A''( \olbx_{1S}\bbeta_S),\ldots, A''( \olbx_{nS}\bbeta_S )\}$ consisting of all positive diagonal entries, because
	of the positivity assumption on $A''(\cdot)$.
	Then for any  column vector $w \in \bR^{p_s+1}$, 
	    $\bV_S^{1/2} \overline{\bX}_S w = 0$ if and only if  $\overline{\bX}_S w = 0$, implying that the positive definiteness of $\nabla^2 \ell_S (\bbeta_S)$ is equivalent to that of $ \widehat{\Sigma}_S = \frac{1}{n} {\overline{\bX}_S}^\rT {\overline{\bX}_S}$. On the other hand, with $p_s\log p_s/n \rightarrow 0$,  Lemma 1 of \cite{fei2019drawing} implies that, 	with probability going to 1, $\|\widehat{\Sigma}_S - \Sigma_S \| \le \varepsilon$ for $\varepsilon = \min(1/2, c_{\min}/2)$, and, hence,  
	\begin{equation*}
	    \lambda_{\min}(\widehat{\Sigma}_S) \ge \lambda_{\min}(\Sigma_S) -\varepsilon \ge \lambda_{\min}(\Sigma) -\varepsilon \ge c_{\min}/2 > 0.
	\end{equation*}

	Thus, with probability going to 1, $\widehat{\Sigma}_S$ is positive definite, yielding that 
	$$\ell_S(\bbeta_S) = n^{-1} \sum_{i=1}^{n}  \left\{  A( \olbx_{iS}\bbeta_S ) - Y_i \olbx_{iS}\bbeta_S  \right\}$$ is strongly convex with respect to $\bbeta_S$.
		Hence, $\widetilde{\bbeta}_{S_{+j} } \in \{\bbeta_S: \ell_S(\bbeta_S)\le A(0) \} $, which is a strongly convex set with probability going to 1. 
As  $A(0) $ does not depend on $S$ or the data, there exists a constant $C_\beta>0$ (which only depends on $A(0)$, but does not depend on $S$ or the data), such that $|\widetilde{\bbeta}_{S }|_\infty \le C_\beta$ holds with probability going to 1.
\end{proof}

\begin{proof}of Theorem \ref{mainthm}:
	
	We define the \textit{oracle} estimators of $\beta^*_{j}$ on the full data $(\bY,\bX)$ and the $b$-th subsample $\bD_1^b$ respectively, where the candidate set is the true set $S^*$:
	\begin{gather*}\label{orc}
	\check{\bbeta}_{S^*_{+j}} = \argmin_{\bbeta \in \bR^{s_0+1} } \ell_{S^*_{+j}}(\bbeta_{S^*_{+j}}) = \argmin_{\bbeta \in \bR^{s_0+1} } \ell_{S^*_{+j}}(\bbeta_{S^*_{+j}}; \bY,\bX_{S^*_{+j}}), \;
	\check{\beta}_{j} = \left(\check{\bbeta}_{S^*_{+j}} \right)_j; \\
	\check{\bbeta}^b_{S^*_{+j}} = \argmin_{\bbeta \in \bR^{s_0+1} } \ell^b_{S^*_{+j}}(\bbeta_{S^*_{+j}})= \argmin_{\bbeta \in \bR^{s_0+1} } \ell_{S^*_{+j}}(\bbeta_{S^*_{+j}}; \bY^{1(b)},\bX^{1(b)}_{S^*_{+j}}), \;
	\check{\beta}^b_{j} = \left(\check{\bbeta}^b_{S^*_{+j}} \right)_j. \label{orc2}
	\end{gather*}
	By Theorem \ref{thm1} and given $s_0^2\log s_0/n \rightarrow 0$, for each $j\in\{1, \ldots,p\}$,
	\begin{equation}\label{clt1}
	\sqrt{n}(\check{\beta}_{j}-\beta^*_j)/\check{\sigma}_j \xrightarrow{d} N(0,1) \quad\mathrm{as}\quad n \rightarrow \infty,
	\end{equation}
	where $\check{\sigma}_j^2 = \left( \{ {I}^*_{S^*_{+j}}\}^{-1}  \right)_{jj}$. 
	
	With the oracle estimators $\check{\beta}_{j}$'s and $\check{\beta}^b_{j}$'s, we have the following decomposition:
	\begin{equation}\label{eq:5}
	\begin{aligned}
	& \sqrt{n}\left(\widehat{\beta}_j - \beta^*_j \right) \\
	=&  \sqrt{n}\left(\check{\beta}_j  - \beta^*_j \right) + \sqrt{n}\left(\widehat{\beta}_j - \check{\beta}_j  \right)  \\
	=& \sqrt{n}\left(\check{\beta}_j  - \beta^*_j \right) + \sqrt{n}\left(\frac{1}{B} \sum_{b=1}^{B} \widetilde{\beta}_j^b - \check{\beta}_j  \right)  \\
	=& \underbrace{\sqrt{n}\left(\check{\beta}_j  - \beta^*_j \right)}_{\rI}  + \underbrace{\sqrt{n}\left(\frac{1}{B} \sum_{b=1}^{B} \check{\beta}_j^b - \check{\beta}_j  \right)}_{\rI\rI} + \underbrace{\sqrt{n}\left(\frac{1}{B} \sum_{b=1}^{B} \left\{\widetilde{\beta}_j^b - \check{\beta}^b_j \right\}  \right)}_{\rI\rI\rI}.
	\end{aligned}
	\end{equation}
	
	The first two terms in (\ref{eq:5}), which do not involve various selections $S^b$'s, deal with the oracle estimators and the true active set $S^*$. We need to show the following,	which will lead to the results stated in the theorem by using Slutsky's theorem. 
	\begin{itemize}
		\item[(a)] $\rI/\check{\sigma}_j = \sqrt{n}\left(\check{\beta}_j  - \beta^*_j \right)/\check{\sigma}_j \stackrel{d}{\rightarrow} N(0,1)$;
		\item[(b)] $\rI\rI =  \frac{\sqrt{n}}{B} \sum_{b=1}^{B} \left\{ \check{\beta}_j^b - \check{\beta}_j \right\} =o_p(1)$;
		\item[(c)] $\rI\rI\rI = \frac{\sqrt{n}}{B} \sum_{b=1}^{B} \left\{\widetilde{\beta}_j^b - \check{\beta}^b_j \right\}   =o_p(1)$.
	\end{itemize}
	 
	 First, (a) holds because of (\ref{clt1}).
	To show (b), i.e. $\rI\rI = o_p(1)$, we first denote $\xi_{b,n} = \sqrt{n}\left( \check{\beta}_j^b - \check{\beta}_j  \right)$, then $\rI\rI = \left(\sum_{b=1}^{B} \xi_{b,n}\right)/B$. Since the sampling indicator vectors, $\bJ_b$'s (defined in Section \ref{s4}) are i.i.d, 
	 $\xi_{b,n}$'s are i.i.d conditional on data 
	 $\bD^{(n)}= (\bY, \bX)$. The conditional distribution of $\sqrt{n}\left( \check{\beta}_j^b - \check{\beta}_j  \right)$ given $\bD^{(n)}$ is asymptotically the same as the unconditional distribution of $ \sqrt{n}\left(\check{\beta}_j  - \beta^*_j \right)$, which converges to zero Gaussian by (\ref{clt1}). 
	 With the uniform boundedness of $\check{\beta}_j^b$ and $\check{\beta}_j$ as shown in Lemma \ref{l1}, 
	we can show that $\bE(\xi_{b,n} |\bD^{(n)})  \rightarrow 0$ and $\var(\xi_{b,n} |\bD^{(n)}) \rightarrow \check{\sigma}_j^2 $ uniformly over $\bD^{(n)}$ as $n\rightarrow \infty$. Furthermore, $\bE(\rI\rI|\bD^{(n)})=\bE(\xi_{b,n} |\bD^{(n)})$,  and $
	\var(\rI\rI |\bD^{(n)}) = \var(\xi_{b,n} |\bD^{(n)})/B.$
	Denote by $\Omega_n$ the sample space of $\bD^{(n)}$. For any $ \delta,\zeta>0$, there exist $N_0,B_0>0$ such that when $n>N_0, B>B_0$,
	\begin{equation*}
	\begin{aligned}
	& \pr(|\rI\rI|\ge \delta)	\le \int_{ \Omega_n}\pr\left(|\rI\rI|\ge \delta\,\middle\vert\, \bD^{(n)}\right)\mathrm{d}\pr(\bD^{(n)}) \\
	\le& \int_{\Omega_n}\pr\left(|\rI\rI - \bE (\rI\rI |\bD^{(n)}) |\ge \delta/2 \,\middle\vert\, \bD^{(n)}\right)\mathrm{d}\pr(\bD^{(n)}) \\
	\le& \int_{\Omega_n}\frac{\var\left(\rI\rI\,\middle\vert\, \bD^{(n)}\right)}{\delta^2/4} \mathrm{d}\pr(\bD^{(n)}) \le\frac{\check{\sigma}_j^2}{B_0\delta^2/4} \int_{\Omega_n}\mathrm{d} \pr(\bD^{(n)}) \le \zeta.
	\end{aligned}
	\end{equation*}
	Thus, $\rI\rI =o_p(1)$.
	
	To prove (c), i.e. $\rI\rI\rI=o_p(1)$, we first note that each subsample $\bD_{1}^b$ can be regarded as a random sample of $n_1= q n $ ($0<q<1$) i.i.d. observations from the population distribution for which Assumption \ref{a3} holds, that is $|S^b|\le K_1n^{c_1} $ and $\pr\left(S^*\subseteq S^b \right)\ge 1 - K_2(p\vee n)^{-1-c_2}$. 
	We show that for any $b$, conditional on $S^b\supseteq S^*$, $\sqrt{n}  \left(\widetilde{\beta}_j^b - \check{\beta}^b_j \right) = o_p(1)$.

	To see this, we first arrange the order of the components of $\bx= (x_1, \ldots, x_p)$ such that the first $s_0$ components are signal variables. In other words, $S^*=\{1, \ldots, s_0\}$.  From (\ref{junk2}) in in the proof of Theorem 1 and omitting superscript $b$, we have that
\begin{eqnarray} \label{junk}
        \widetilde{\beta}_{j} - \beta^*_{j} & = &  -n_1^{-1} \widetilde{e}_j^\rT  \{I^*_{S_{+j}} \}^{-1} U_{S_{+j}}(\bbeta^*_{S_{+j}}) + \widetilde{r}_{n_1}, \nonumber\\
        \check{\beta}_{j} - \beta^*_{j} & = &  -n_1^{-1} \check{e}_j^\rT \{I^*_{S^*_{+j}} \}^{-1} U_{S^*_{+j}}(\bbeta^*_{S^*_{+j}}) + \check{r}_{n_1},
\end{eqnarray}
where $\widetilde{e}_j=(0,\ldots,0,1, 0, \ldots,0)^\rT$
 is a unit vector of length $|{S_{+j}} |$
to index the position of variable $j$ in $S_{+j}$, $\check{e}_j$ is a unit vector of length $|{S^*_{+j}} |$  to index the position of variable $j$ in $S^*_{+j}$, and the residuals $\|\widetilde{r}_{n_1}\|^2_2 = o_p(1/n_1), \|\check{r}_{n_1}\|^2_2 = o_p(1/n_1)$. 
Here, $I^*_{S_{+j}}$ and  $I^*_{S^*_{+j}}$ are two submatrices of the expected information at $\bbeta^*$, i.e. $I^* =  \bE \{\frac{1}{n}{\overline{\bX}}^\rT \bV \overline{\bX}\}$, 
 where $\bV = \mathrm{diag}\{\nu(\mu_1),\ldots,\nu(\mu_n)\}$ is an $n \times n$ diagonal matrix with $\mu_i= g^{-1}(\overline{\bx}_i\bbeta^*)$; see Section 2.1 for the notation.

Therefore, the $jk$-th $(j,k=0,1, \ldots, p)$ entry of $I^*$, a $(p+1) \times (p+1)$ matrix, is $\bE(x_j \nu(\mu) x_k)$ with $\mu= g^{-1}(\overline{\bx}\bbeta^*)$.
Now for any $j \in S^*$, $k \in S^c$, the complement of $S^*$, the  partial orthogonality condition (Fan and Lv, 2008; Fan and Song, 2010) that  $\{x_j, j \in S^*\}$ are independent of $\{x_k, k \in S^c\}$ implies that
$\bE(x_j \nu(\mu) x_k) =0$,
as $\mu$ only depends on  $x_j', j' \in S^*$ and $\bE(x_k)=0$ with centered predictors. Therefore, $I^*$ is block-diagonal with two blocks indexed by $S^*$ and $S^c$. That is, 
\begin{equation*}
   I^*= 
\left(\begin{array}{cc}
\bE ( \frac{1}{n} \overline{\bX}_{S^*}^\rT\bV\overline{\bX}_{S^*}) &  0 \\
 0 &  \bE ( \frac{1}{n}{\bX}_{S^c}^\rT\bV {\bX}_{S^c})
\end{array}\right).
\end{equation*}
where the submatrices $\overline{\bX}_{S^*}$ and ${\bX}_{S^c}$ are as defined in
Section 2.1.
 Hence, ${I}^*_{S_{+j}}$ is block-diagonal with two blocks indexed by $S^*$ and $S_{+j}\setminus S^*$, and ${I}^*_{S^*_{+j}}$ is block-diagonal with two blocks indexed by $S^*$ and $S^*_{+j}\setminus S^*=\emptyset$ if $j\in S^*$ or $=\{j\}$ otherwise.
 
Therefore,  $\{I^*\}^{-1}$, $ \{{I}^*_{S_{+j}} \}^{-1}$ and $\{{I}^*_{S^*_{+j}} \}^{-1}$ are all block-diagonal. Furthermore, the blocks corresponding to $S^*$ in $ \{{I}^*_{S_{+j}} \}^{-1}$ and $\{{I}^*_{S^*_{+j}} \}^{-1}$ are identical and are equal to 
$\{\bE ( \frac{1}{n} \overline{\bX}_{S^*}^\rT\bV \overline{\bX}_{S^*})\}^{-1}$.
Write $U(\bbeta^*) = (u_0, u_1, u_2, \ldots, u_p)^\rT$, 
$\widetilde{e}_j^\rT \{{I}^*_{S_{+j}} \}^{-1}=  (\widetilde{i}_{jk})_{k \in S_{+j}}$ and $\check{e}_j^\rT  \{{I}^*_{S^*_{+j}} \}^{-1}  = (\check{i}_{jk})_{k \in S^*_{+j}}$. Then, it follows that $\widetilde{i}_{jk} = \check{i}_{jk}$ for $k\in S^*$, which leads to 
\[
    \begin{aligned}
    \sqrt{n_1}\left(\widetilde{\beta}_{j} - \check{\beta}_{j} \right) =& - \frac{1}{\sqrt{n_1}}\sum_{k \in S_{+j}}\widetilde{i}_{jk}u_k + \frac{1}{\sqrt{n_1}}\sum_{k \in S^*_{+j}}\check{i}_{jk}u_k + r'_{n_1} \\
     =& - \frac{1}{\sqrt{n_1}}\sum_{k \in S\setminus S^* }\widetilde{i}_{jk}u_k  + \frac{1 (j \notin S^*)}{\sqrt{n_1}}
     \left(\check{i}_{jj} - \widetilde{i}_{jj} \right)u_j + r'_{n_1} 
    \end{aligned}
\]
where $ r'_{n_1} = \sqrt{n_1} (\widetilde{r}_{n_1} -\check{r}_{n_1}) = o_p(1)$,
and $\widetilde{r}_{n_1}$ and
$\check{r}_{n_1}$ are as in
(\ref{junk}). 

With Assumption (A3), $|S\setminus S^*| \le K_1n^{c_1} = o(\sqrt{n_1})$ with $ 0 \le c_1<1/2$ and, thus,  $\var(\sum_{k \in S\setminus S^* }\widetilde{i}_{jk}u_k ) = o(n_1)$. By the Chebyshev inequality,
the first term on the right hand side converges to 0 in probability. Thus, 
each of these three terms  is  $o_p(1)$ and we have $\sqrt{n_1}\left(\widetilde{\beta}_{j} - \check{\beta}_{j} \right) = o_p(1)$. As $n_1/n = q$ where $ 0<q<1$, the original statement holds.


	Now define $\eta_b= 1\left\{ S^* \not\subseteq S^b \right\}\sqrt{n}  \left(\widetilde{\beta}_j^b - \check{\beta}^b_j \right)$, while omitting subscripts $j$ in $\eta$ for simplicity, then $\rI\rI\rI =  \left(\sum_{b=1}^{B} \eta_b \right)/B$.
	When $S^* \not\subseteq S^b $, $\widetilde{\beta}_j^b $ is not an unbiased estimator of  $\beta^*_j$ any more, instead we try to bound it by some constant. By Lemma \ref{l1}, there exists $C_\beta\ge c_\beta$ such that $\sup_{b} \left|\widetilde{\beta}_j^b - \check{\beta}^b_j \right| \le \sup_{b} \left|\widetilde{\beta}_j^b - {\beta}^*_j \right| + \sup_{b} \left|\check{\beta}_j^b - {\beta}^*_j \right|  \le 2C_\beta +1$. Therefore, by \ref{a3},
	\begin{eqnarray} \label{junk3}
	\bE(\eta_b) &\le & \pr\left( S^* \not\subseteq S^b \right) \sqrt{n} \sup_{ 1 \le b \le  B} \left|\widetilde{\beta}_j^b - \check{\beta}^b_j \right| \le 2C_\beta\sqrt{n} K_2(p\vee n)^{-1-c_2},  \nonumber \\
	\var(\eta_b) & \le &  \pr\left(S^* \not\subseteq S^b \right) n\sup_{1 \le b \le  B} \left(\widetilde{\beta}_j^b - \check{\beta}^b_j \right)^2 \le 4C_\beta^2n K_2(p\vee n)^{-1-c_2}.
	\end{eqnarray}
	With dependent $\eta_b$'s, we further have
	\begin{gather*}
	\bE(\rI\rI\rI) = \bE\left\{\left(\sum_{b=1}^{B} \eta_b \right)/B  \right\} \le 2C_\beta\sqrt{n} K_2(p\vee n)^{-1-c_2}, \\
	\var(\rI\rI\rI) \le \frac{1}{B^2} \sum_{b=1}^{B}\sum_{b'=1}^{B}|\cov\left(\eta_b,\eta_{b'}\right)| \le 4C_\beta^2n K_2(p\vee n)^{-1-c_2},
	\end{gather*} 
	where the last inequality holds because of  $|\cov\left(\eta_b,\eta_{b'}\right)|
\le \{ \var(\eta_b)
\var(\eta_{b'})\}^{1/2}$ and (\ref{junk3}).
	Then we show $\rI\rI\rI = o_p(1)$. More specifically, for any $\delta>0, \zeta>0$, take $N_0=\lfloor (C_\beta \delta)^{1/2+c_2}\rfloor$, where,
 for a real number $a$,	$\lfloor a \rfloor$ denotes the integer part of $a$. When $n>N_0$, $\bE (\rI\rI\rI) \le \delta/2$. Also let $N_1 = \lfloor \{\zeta \delta^2/(16C_\beta^2 K_2)\}^{c_2}\rfloor$. Then when $n > \max(N_0, N_1)$, we have
	\begin{eqnarray*}
	\pr(|\rI\rI\rI|\ge \delta)	& \le &  \pr\left(|\rI\rI\rI - \bE (\rI\rI\rI) |\ge \delta/2 \right) \\
	& \le  & \frac{\var\left(\rI\rI\rI\right)}{\delta^2/4} \le \frac{16C_\beta^2 K_2}{\delta^2} n (p\vee n)^{-1-c_2} \\ 
	& < & 
	\frac{16C_\beta^2 K_2}{\delta^2} n^{-c_2}
	< \zeta,
	\end{eqnarray*}
	where the first inequality is due to $|\bE(\rI\rI\rI)| \le \delta/2$ when $n>N_0$,  the second one is due to the Chebyshev inequality and the last one is due to $n>N_1$.
\end{proof}

\begin{proof}of Theorem \ref{thm3}:

	Following the previous proof, we replace the arguments in $j$ with those in $S^{(1)}$. The \textit{oracle} estimators are
	\begin{gather*}
	\check{\bbeta}_{S^{(1)}\cup S^*}  = \argmin\ell_{S^{(1)}\cup S^*}(\bbeta_{S^{(1)}\cup S^*}; \bY,\bX_{S^{(1)}\cup S^*}), \;
	\check{\bbeta}_{S^{(1)}} = \left( \check{\bbeta}_{S^{(1)}\cup S^*} \right)_{S^{(1)}}; \\
	\check{\bbeta}^b_{S^{(1)}\cup S^*} = \argmin\ell_{S^{(1)}\cup S^*}(\bbeta_{S^{(1)}\cup S^*}; \bY^{1(b)},\bX^{1(b)}_{S^{(1)}\cup S^*}), \;
	\check{\bbeta}^b_{S^{(1)}} = \left( \check{\bbeta}^b_{S^{(1)}\cup S^*} \right)_{S^{(1)}}.
	\end{gather*}
	Notice that $|S^{(1)}|=p_1=O(1)$, as $n\rightarrow \infty$, $|S^*\cup {S^{(1)}}|=O\big(|S^*| \big)=o(n)$, so that the above quantities are well-defined. The oracle estimator  follows
	\begin{equation*}\label{2clt1}
	\sqrt{n}\left\{ ({I}^*_{S^{(1)} \cup S^*} )^{-1/2} \right\}_{S^{(1)}} \left(\check{\bbeta}_{S^{(1)}}-\bbeta^*_{S^{(1)}}\right) \xrightarrow{d} N(0,\bI_{p_1}) \quad\mathrm{as}\quad n \rightarrow \infty.
	\end{equation*}
	Here, for a square matrix, say, $Q$, $(Q)_{S}$ is a submatrix of $Q$ with rows and columns indexed by $S$. 
	Denote by ${I}^{(1)} =\left\{ ({I}^*_{S^{(1)} \cup S^*} )^{-1/2} \right\}_{S^{(1)}}$.
	
	Similar to (\ref{eq:5}), we have a decomposition
	\begin{equation*}
	\sqrt{n} {I}^{(1)} (\widehat{\bbeta}^{(1)} - \bbeta^*_{S^{(1)}})
	= \sqrt{n} {I}^{(1)} \left(\check{\bbeta}_{S^{(1)}}  - \bbeta^*_{S^{(1)}} \right) + \sqrt{n} {I}^{(1)} \left(\frac{1}{B} \sum_{b=1}^{B} \widetilde{\bbeta}_{S^{(1)}}^b - \check{\bbeta}_{S^{(1)}}  \right).  
	\end{equation*}
	Analogous to the derivations in the previous proof, it follows that the second term is $o_p(1)$.
	 Hence, the theorem holds.
\end{proof}

\section*{Appendix B: Additional Simulations}

To assess the robustness of our method, we performed additional simulations when the parametric model was  mis-specified and
when the sparsity condition was violated. 

{\bf Example B.1} assumed that $Y|\bx$ followed a negative binomial distribution:
\begin{gather*}
    \pr(Y = y) = \frac{\Gamma(y+r)}{\Gamma(r)y!}p^r(1-p)^y, \\
    \bE Y = \mu = r(1-p)/p = \bx\bbeta^*,
\end{gather*}
with $r=10$, sample size $n=300$, $p=500$, and $s_0 = 5$. However, we modeled the data using SSGLM under the Poisson regression (\ref{pois}) with  $B=300$. Table \ref{misspec} summarizes the results based on $200$ simulated data sets. The $\widehat{\beta}_j$'s had  small biases. The estimated standard errors were slightly less than the empirical standard deviations. Nevertheless, the coverage probabilities were still close to the $0.95$ nominal level.

\begin{table}[ht]
\centering
\caption{SSGLM under mis-specified model. \label{misspec}}
\begin{tabular}{rrrrrrr}
  \hline
 Index $j$ & 90 & 179 & 206 & 237 & 316 & Noise \\ 
 $\beta^*_j$ & -1.000 & -0.500 & 0.500 & 1.000 & 1.500 & 0.000 \\ 
  \hline
  Bias & -0.020 & 0.020 & 0.018 & 0.001 & 0.010 & -0.001 \\ 
  SE & 0.240 & 0.235 & 0.232 & 0.236 & 0.243 & 0.233 \\ 
  SD & 0.258 & 0.243 & 0.249 & 0.249 & 0.250 & 0.231 \\ 
  Cov prob & 0.955 & 0.945 & 0.900 & 0.930 & 0.925 & 0.946 \\ 
  Sel freq & 0.724 & 0.177 & 0.216 & 0.723 & 0.977 & 0.021 \\ 
   \hline
\end{tabular}
\end{table}

{\bf Example B.2} assumed a non-sparse truth $\bbeta^*$ under the Poisson truth (\ref{pois}). With $n=300$ and $p=500$, we let $s_0 = 100$. Among the 100 predictors with non-zero effects,   $96$ $\beta^*_j$'s were small, which were randomly drawn from $\mathrm{Unif}[-0.5, 0.5]$, and the other 4 had values $-1.5, -1, 1, 1.5$ (as shown in Figure \ref{nonsparse}). With many small but non-zero signals, SSGLM still gave nearly  unbiased estimates to all of them. See Table \ref{tabnonsparse}, where the columns represent 4 large size $\beta^*_j$'s, and the averages over all small signals and all noise variables, respectively.

\begin{figure}
	\raggedleft
	\caption{SSGLM under non-sparse truth, with $p=500$ and $s_0 =100$.\label{nonsparse}}
	\includegraphics[scale=0.7]{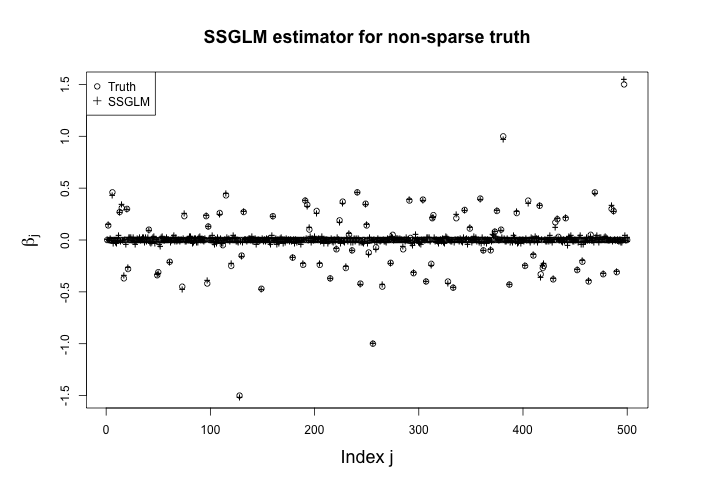}
\end{figure}
\begin{table}[ht]
\centering
\caption{SSGLM under non-sparse truth.\label{tabnonsparse}}
\begin{tabular}{rrrrrrr}
  \hline
Index $j$ & 128 & 256 & 381 & 497 & Small & Noise \\ 
$\beta^*_j$ & -1.50 & -1.00 & 1.00 & 1.50 & - & 0 \\ 
  \hline
  Bias & -0.01 & 0.003 & -0.03 & 0.05 & 0.003 & $9\times 10^{-4}$ \\ 
  SE & 0.31 & 0.30 & 0.30 & 0.31 & 0.29 & 0.30 \\ 
  SD & 0.31 & 0.30 & 0.32 & 0.30 & 0.29 & 0.29 \\ 
  Cov prob & 0.93 & 0.93 & 0.93 & 0.94 & 0.94 & 0.94 \\ 
  Sel freq & 0.87 & 0.50 & 0.49 & 0.87 & 0.06 & 0.03 \\ 
   \hline
\end{tabular}
\end{table}

\end{document}